\documentclass[letterpaper,11pt,english]{article}
\usepackage[letterpaper, portrait, margin=2.5cm]{geometry}

\usepackage{amsmath,amsthm}\allowdisplaybreaks
\usepackage{amssymb}
\usepackage{mathtools}
\usepackage{csquotes}
\usepackage{graphicx}
\usepackage{grffile}
\usepackage{subfig}
\usepackage[hidelinks]{hyperref}

\usepackage[noabbrev,capitalise]{cleveref}

\usepackage{soul}
\usepackage{xcolor}
\usepackage{xspace}
\usepackage{paralist}
 
\usepackage{algorithm}
\usepackage[noend]{algorithmic}
\usepackage{textcomp}
\usepackage{nicefrac}

\usepackage{todonotes}

\usepackage[subtle]{savetrees}

\theoremstyle{plain}
\newtheorem{theorem}{Theorem}
\newtheorem{lemma}[theorem]{Lemma}
\newtheorem{proposition}[theorem]{Proposition}

\newtheorem{corollary}[theorem]{Corollary}
\newtheorem{definition}{Definition}

\usepackage{authblk}

\title{Managing Multiple Mobile Resources (Full Version)
\thanks{This work was partially supported by the German Research Foundation (DFG) within the Collaborative Research Centre "On-The-Fly Computing" under the project number 160364472 --- SFB 901/3.
}}
\author[]{Bj\"orn Feldkord}
\author[]{Till Knollmann}
\author[]{Manuel Malatyali}
\author[]{Friedhelm Meyer auf der Heide}
\affil[]{Heinz Nixdorf Institute and Department of Computer Science\\
	Paderborn University, F\"urstenallee 11, 33102 Paderborn, Germany
}
\affil[]{ \{bjoernf, tillk, malatya, fmadh\}@mail.upb.de}
\date{}

\begin{document}

\maketitle

\begin{abstract}
We extend the Mobile Server Problem introduced in~\cite{DBLP:conf/spaa/FeldkordH17}
to a model where $k$ identical mobile resources, here named servers, answer requests appearing at points in the Euclidean space.
In order to reduce communication costs, the positions of the servers can be adapted by a limited distance $m_s$ per round for each server.
The costs are measured similar to the classical Page Migration Problem, i.e., answering a request induces costs proportional to the distance to the nearest server,
and moving a server induces costs proportional to the distance multiplied with a weight $D$.

We show that, in our model, no online algorithm can have a constant competitive ratio, i.e., one which is independent of the input length $n$, 
even if an augmented moving distance of $(1+\delta)m_s$ is allowed for the online algorithm.
Therefore we investigate a restriction of the power of the adversary dictating the sequence of requests: We demand \emph{locality of requests}, i.e., that consecutive requests come from points in the Euclidean space with distance bounded by some constant $m_c$. 
We show constant lower bounds on the competitiveness in this setting (independent of $n$, but dependent on $k$, $m_s$ and $m_c$).

On the positive side, we present a deterministic online algorithm with bounded competitiveness when augmented moving distance and locality of requests is assumed.
Our algorithm simulates any given algorithm for the classical $k$-Page Migration problem as guidance for its servers and extends it by a greedy move of one server in every round.
The resulting competitive ratio is polynomial in the number of servers $k$, the ratio between $m_c$ and $m_s$, the inverse of the augmentation factor
$\nicefrac{1}{\delta}$ and the competitive ratio of the simulated $k$-Page Migration algorithm.
\end{abstract}

%%%%%%%%%%%%%%%%%%%%%%%%%%%%%%%%%%%%%%%%%%%%%%%%%%%%%%%%%%%%%%%%%%%%%%%%%%%%%%%%%%%%%%%%%%%%%%%%%%%
%%%%%%%%%%%%%%%%%%%%%%%%%%%%%%%%%%%%%%%%%%%%%%%%%%%%%%%%%%%%%%%%%%%%%%%%%%%%%%%%%%%%%%%%%%%%%%%%%%%

\section{Introduction}
We consider a scenario where several devices continuously access a common set of $k$ identical resources. 
The devices pose \emph{requests} for a resource which must be answered by communicating with the resources, incurring cost for communication.
The placement of resources is managed by an algorithm whose goal it is to reduce the costs for communication and for moving the resources as much as possible.
Typically requests do not concern the complete resource (which may be rather large) and it is cheaper to answer a request for a resource by communicating with the resource instead of moving it.
We assume requests appear in an online fashion, i.e., it is unknown to the algorithm where the next requests will arrive while newly arriving requests must be answered instantly to provide latency guarantees.

The scenario described above can be modeled based on the classical Page Migration problem~\cite{Black89}:
A single resource can be moved between two points $a$ and $b$ for costs $D\cdot d(a,b)$, where $d(a,b)$ is the distance between $a$ and $b$ and $D\geq 1$ is a constant.
In every round a request appears at some point $r$, and if the current position of the resource is $p$, it is served for costs $d(p,r)$.
This problem was extended to the Mobile Server Problem~\cite{DBLP:conf/spaa/FeldkordH17}, which puts a limit on how much the resource (called server) can move in each time step and therefore introduces the idea that the positions of resources can not be arbitrarily changed in each time step.

In our work, we extend this idea to multiple resources:
We consider $k$ identical servers located in the Euclidean space (of arbitrary dimension).
Each of them may move a distance of at most $m_s$ per time step.
In each time step, a request appears which has to be served by one of the servers by the end of the time step.
The cost function is the same as in the Page Migration Problem.

\subsection{Related Work}
Besides being a direct extension of the Mobile Server Problem~\cite{DBLP:conf/spaa/FeldkordH17}, our work builds on and is related to results surrounding the $k$-Server and Page Migration problems.
These problems have been examined in many variants and especially for the $k$-Server Problem there are many algorithms for special metrics.
In this overview we only focus on most relevant results for our problem, which are mostly algorithms with an (asymptotically) optimal competitive ratio.

In the classical $k$-Server Problem as introduced by Manasse et al.~\cite{DBLP:journals/jal/ManasseMS90}, $k$ identical servers are located in a metric space and requests are answered by moving at least one of the servers to the point of the request.
The associated costs are equal to the total distance moved.
Manasse et al.\ showed that no online algorithm could be better than $k$-competitive on every metric with at least $k+1$ points.
They stated as the $k$-Server Conjecture that there is a $k$-competitive online algorithm for every metric space.
Further, the conjecture is shown to hold for $k=2$ and $k=n-1$ where $n$ is the number of points in the metric space.

Since its introduction, many algorithms have been designed for special cases of the problem.
Most notable is the Double-Coverage Algorithm~\cite{DBLP:journals/siamdm/ChrobakKPV91}, which is $k$-competitive on trees.
For general metrics, the best known result is the Work-Function Algorithm, which is shown to be $2k-1$-competitive~\cite{DBLP:journals/jacm/KoutsoupiasP95}.
Although this algorithm seems generally inefficient in case of runtime and memory, there have been studies showing that an efficient implementation of this algorithm is indeed possible~\cite{DBLP:journals/cejor/RudecBM13, DBLP:journals/cejor/RudecM15}.
It was also shown that the algorithm has an optimal competitive ratio of $k$ on line and star metrics, as well as metrics with $k+2$ points~\cite{DBLP:journals/tcs/BartalK04}.

The study of randomized online algorithms was initiated by Fiat et al.~\cite{DBLP:journals/jal/FiatKLMSY91} who gave a $\log (k)$-competitive algorithm for the complete graph.
It is speculated that this factor can be obtained for all metrics, however the question is still open.
For general metrics, the first algorithm with polylogarithmic competitive ratio was an $\mathcal{O}(\log^3 n\cdot\log^2 k)$-competitive algorithm by Bansal et al.~\cite{DBLP:journals/jacm/BansalBMN15}.
This was recently improved by Bubeck et al.~\cite{DBLP:conf/stoc/BubeckCLLM18} who gave an $\mathcal{O}(\log^2 k)$-competitive algorithm for HSTs
which can be turned into an $\mathcal{O}(\log^{9}(k)\cdot\log\ \log (k))$-competitive one for general metrics by a dynamic embedding of general metrics into HSTs~\cite{DBLP:conf/focs/Lee18}.

Regarding the Page Migration Problem~\cite{Black89} (also known as File Migration Problem), most results focus on online algorithms which handle only a single page.
Contrary to the $k$-Server Problem, the design of such algorithms is not trivial for the Page Migration Problem.
To the best of our knowledge, the current best results are a $4$-competitive deterministic algorithm by Bienkowski et al.~\cite{DBLP:conf/icalp/BienkowskiBM17} and a collection of randomized algorithms with competitive ratio of at most 3 by Jeffery Westbrook~\cite{DBLP:journals/siamcomp/Westbrook94}.
The most relevant results for our problem are two constructions by Bartal et al.~\cite{DBLP:journals/tcs/BartalCI01} who give both a deterministic and a randomized algorithm which transform a given algorithm for the $k$-Server problem into a deterministic / randomized algorithm for the $k$-Page Migration Problem. If the given $k$-Server algorithm is $c$-competitive, the deterministic algorithm is $\mathcal{O}(c^2)$-competitive, the randomized algorithm is $\mathcal{O}(c)$-competitive.
Conversely, we use the resulting algorithms as a black box in our constructions.

\subsection{Our Results \& Outline of the Paper}
In~\cite{DBLP:conf/spaa/FeldkordH17} it was already shown that no online algorithm for our problem can be competitive even on the real line and with just $k=1$ server.
As a consequence, we employ the following methods to derive reasonable results for the problem:
On the one hand we restrict the adversary to the case with \emph{locality of requests}, i.e., we introduce a parameter $m_c$ by which we can define families of instances classified by the maximum distance between two consecutive requests.
On the other hand we apply \emph{resource augmentation} as in~\cite{DBLP:conf/spaa/FeldkordH17}, i.e., we allow the online algorithm to use a maximum movement distance of $(1+\delta)m_s$.
We show that, for $k\geq 2$, both methods are needed to yield competitive bounds independent of the length of the instance. 
For $k=1$, it was shown in~\cite{MSParxiv} that a locality of requests can improve the competitiveness, but is not necessary to achieve a constant upper bound.

The parameters $m_c$ and $m_s$ have a crucial impact on the resulting competitiveness and thus separate simple from hard instances.
We are able to show that these parameters seem to naturally describe the problem, since we can prove a lower bound of $\Omega(\frac{m_c}{m_s})$.
For fast moving resources ($m_c < (1+\delta)m_s$), our algorithm has an almost optimal competitive ratio when given an optimal $k$-Page Migration algorithm.
For the case of slow moving resources ($m_c \geq (1+\delta)m_s$), we can achieve bounds independent of the length of the input stream.
In detail, we obtain a bound of $\mathcal{O}(\frac{1}{\delta^4}\cdot k^2\cdot \frac{m_c}{ m_s} +\frac{1}{\delta^3}\cdot k^2 \cdot\frac{m_c}{ m_s}\cdot c(\mathcal{K}))$, where $c(\mathcal{K})$ is the competitiveness of a given $k$-Page Migration algorithm.
For the case $D=1$, which we call the \emph{unweighted problem}, the $k$-Page Migration algorithm can be replaced by a $k$-Server algorithm.
Note that the parameter $\varepsilon$ in Table~\ref{tab:results} is indirectly given as the relative difference between $m_c$ and $m_s$.
If $m_c<m_s$, then in the first row we have $\varepsilon>\delta$.
Alternatively, if $\delta=0$, this case still yields an almost optimal upper bound up to a factor of $\nicefrac{1}{\varepsilon}$.

\begin{table}[htbp]
	\centering
		\begin{tabular}{|r|c|c|c|}
\cline{1-4}
 & Lower bound & Unweighted Problem ($D=1$) & Weighted Problem ($D>1$) \\
\cline{1-4} 
$m_c\leq (1+\delta-\varepsilon)m_s$ & $\Omega(k)$ & $\mathcal{O}(\frac{1}{\varepsilon}\cdot k)$ & $\mathcal{O}(\frac{1}{\varepsilon}\cdot k^2)$\\ 
\cline{1-4}
$m_c\geq (1+\delta)m_s$ & $\Omega(k+\frac{m_c}{m_s})$ & $\mathcal{O}(\frac{1}{\delta^4}\cdot k^{3}\cdot \frac{m_c}{m_s})$ & $\mathcal{O}(\frac{1}{\delta^4}\cdot k^{4}\cdot \frac{m_c}{m_s})$ \\
\cline{1-4}
		\end{tabular}
	\caption{An overview of the results, using the best known deterministic algorithms for $k$-Server / $k$-Page Migration.
	  The results in the first row also hold without resource augmentation when $m_c\leq (1-\varepsilon)m_s$.}
	\label{tab:results}
\end{table}

The paper is structured as follows:
A formal definition of our model can be found in \cref{sec:model}.
All relevant lower bounds are established in \cref{sec:lower-bounds}.
In terms of upper bounds, we first give an algorithm for the \emph{unweighted problem} in Section~\ref{sec:unweighted}.
The analysis for instances with $m_c<(1+\delta)m_s$ consists of a simple potential function argument found in Section~\ref{sec:trivialunweighted}.
The analysis of the other case is much more challenging and is conducted in Section~\ref{sec:projection}.
The weighted case ($D>1$) is discussed in \cref{sec:weighted}.
While the basic approach stays the same, we need to modify the movement of the online algorithm due to the higher movement costs.
We show how the algorithm can be adapted and present the resulting competitive ratio following a similar structure as in the unweighted case.

%%%%%%%%%%%%%%%%%%%%%%%%%%%%%%%%%%%%%%%%%%%%%%%%%%%%%%%%%%%%%%%%%%%%%%%%%%%%%%%%%%%%%%%%%%%%%%%%%%%
%%%%%%%%%%%%%%%%%%%%%%%%%%%%%%%%%%%%%%%%%%%%%%%%%%%%%%%%%%%%%%%%%%%%%%%%%%%%%%%%%%%%%%%%%%%%%%%%%%%

\section{Model \& Notation}
\label{sec:model}

In this section we formally describe the model and some common notation used throughout the paper.

Time is considered discrete and divided into time steps $1,\ldots, n$.
An input to the $k$-Mobile Server Problem is given by a sequence of requests $r_1,\ldots,r_n$ where each $r_t$ occurs in time step $t$ and is represented by a point in the Euclidean space of arbitrary dimension.
We are given $k$ servers $a_1,\ldots,a_k$ controlled by our online algorithm.
At each point in time, one server occupies exactly one point in the Euclidean space.
We denote by $a_i^{(t)}$ the position of server $a_i$ at end of time step $t$, and by $d(a,b)$ the Euclidean distance between two points $a$ and $b$.
For the distance between two servers $a_i^{(t)}$ and $a_j^{(t)}$ in the same time step $t$, we also use the notation $d_t(a_i,a_j)$.
We may also leave out the time $t$ entirely if it is clear from the context.

In each time step $t$, the current request $r_t$ is revealed to the online algorithm.
The algorithm may then move each server, such that $d(a_i^{(t-1)},a_i^{(t)})\leq m_s$ for all servers $a_i$.
The movement incurs cost of $D\cdot\sum_{i=1}^{k}d(a_i^{(t-1)},a_i^{(t)})$ for a constant $D\geq 1$.
The request $r_t$ is then served by the closest server $a_i^{(t)}$, which incurs cost of $d(a_i^{(t)},\ r_t)$.
Note that the variables indexed with the time $t$ represent the configuration at the end of the time step $t$.

In our model, we consider the locality of requests dictated by a parameter $m_c$ limiting the distance between consecutive requests, i.e., $d(r_t,r_{t+1})\leq m_c$.
We also consider a resource augmentation setting, where the maximum distance an online algorithm may move is in fact $(1+\delta)m_s$ for some $\delta\in(0,1)$.
The cost of our online algorithm is denoted by $C_{Alg}$.
We compare the costs of an online algorithm to an offline optimum, whose servers are denoted by $o_1,\ldots,o_k$ and whose cost is $C_{Opt}$.

%%%%%%%%%%%%%%%%%%%%%%%%%%%%%%%%%%%%%%%%%%%%%%%%%%%%%%%%%%%%%%%%%%%%%%%%%%%%%%%%%%%%%%%%%%%%%%%%%%%
%%%%%%%%%%%%%%%%%%%%%%%%%%%%%%%%%%%%%%%%%%%%%%%%%%%%%%%%%%%%%%%%%%%%%%%%%%%%%%%%%%%%%%%%%%%%%%%%%%%

\section{Lower Bounds}\label{sec:lower-bounds}

In this section, we will prove lower bounds for the competitive ratio of our problem. They show the importance both of the resource augmentation and the locality of requests introduced above.
All our lower bounds already hold on the line (and therefore in arbitrary dimensions, too).
Since our model is an extension of the $k$-Page Migration Problem, $\Omega(k)$ is a lower bound for deterministic algorithms inherited from that problem (which itself inherits the bound from the $k$-Server Problem, see \cite{DBLP:journals/tcs/BartalCI01,DBLP:journals/jal/ManasseMS90}).
Even when $m_c$ is restricted, the lower bound instance can simply be scaled down such that the distance limits are not relevant for the instance.

We start by discussing the model without any restriction on the distance between the requests in two consecutive time steps, i.e., the parameter $m_c$ is unbounded.
We also consider the case, that there is no resource augmentation, i.e., the maximum movement distance of the online algorithm and of the offline solution are the same.
The following lower bound, originally formulated for $k=1$, carries over from~\cite{DBLP:conf/spaa/FeldkordH17}:

\begin{theorem}
Every randomized online algorithm for the Mobile Server Problem (with $k=1$) has a competitive ratio of $\Omega(\frac{\sqrt{n}}{D})$ against an oblivious adversary,
where $n$ is the length of the input sequence.
\end{theorem}

For more than one server, we obtain an additional bound which can not be resolved with the help of resource augmentation.
The proofs of the following theorems can be found in Appendix~\ref{app:lower}.

\begin{theorem}\label{th:lowerunbounded}
For $k\geq 2$, every randomized online algorithm for the $k$-Mobile Server Problem has a competitive ratio of at least $\Omega(\frac{n}{Dk^2})$,
where $n$ is the length of the input sequence.
\end{theorem}

Since we often consider input sequences for problems such as ours to be potentially infinite, we deem competitive ratios dependent on the input length undesirable.
Hence, as a consequence of the bounds shown so far, we apply two modifications to our model which help us to achieve a competitive ratio independent of the length of the input sequence.
We use the concept of resource augmentation just as in~\cite{DBLP:conf/spaa/FeldkordH17} to allow the online algorithm to utilize a maximum movement distance of $(1+\delta)m_s$ for some $\delta\in(0,1)$ as opposed to the distance $m_s$ used by the optimal offline solution.
This measure alone does not address the bound from Theorem~\ref{th:lowerunbounded} (the ratio shrinks, but still depends on $n$).
Hence, we introduce the locality of requests, i.e., restrict the distance between two consecutive requests to a maximum distance of $m_c$.
Note, that only restricting the distance between consecutive requests does also not remove the dependence on $n$, as was shown in~\cite{MSParxiv}.
The following theorem can be obtained in a similar way as Theorem~\ref{th:lowerunbounded}:

\begin{theorem}
\label{th:restrictedlower}
	For $k\geq 2$, every randomized online algorithm for the \( k \)-Mobile Server Problem, where the distance between consecutive requests is bounded by \( m_{c} \), has a competitive ratio of at least \( \Omega (\frac{m_{c}}{m_{s}}) \).
\end{theorem}

%%%%%%%%%%%%%%%%%%%%%%%%%%%%%%%%%%%%%%%%%%%%%%%%%%%%%%%%%%%%%%%%%%%%%%%%%%%%%%%%%%%%%%%%%%%%%%%%%%%
%%%%%%%%%%%%%%%%%%%%%%%%%%%%%%%%%%%%%%%%%%%%%%%%%%%%%%%%%%%%%%%%%%%%%%%%%%%%%%%%%%%%%%%%%%%%%%%%%%%

\section{An Algorithm for the Unweighted Problem}
\label{sec:unweighted}
In this section we consider the unweighted problem ($D=1$).
Our algorithm does the following:
We mainly follow around a simulated $k$-Server algorithm, but always move the closest server greedily towards the request.

We use the following notation in this section:
Denote by $a_1,\ldots,a_k$ the servers of the online algorithm, $c_1,\ldots,c_k$ the servers of the simulated $k$-server algorithm and $o_1,\ldots,o_k$ the servers of the optimal solution.
For an offline server $o_i$, we denote by $o_i^{a}$ the closest server of the online algorithm to $o_i$ (this might be the same server for multiple offline servers).
Furthermore, we denote by $a^{*}$, $c^{*}$ and $o^{*}$ the closest server to the request of the algorithm, the $k$-server algorithm, and the optimal solution respectively.
For a fixed time step $t$, we add a $"'"$ to any variable to denote the state at the end of the current time step, e.g.,\ $a_1=a_1^{t-1}$ is the position of the server at the beginning of the time step and $a_1'=a_1^{t}$ is the position at the end of the current step.

\medskip

\noindent
Our algorithm \emph{Unweighted-Mobile Servers (UMS)} works as follows:\\
Take any $k$-Server algorithm $\mathcal{K}$ with bounded competitiveness in the Euclidean space.
Upon receiving the next request $r'$, simulate the next step of $\mathcal{K}$.
Calculate a minimum weight matching (with the distances as weights) between the servers $a_1,\ldots,a_k$ of the online algorithm and the servers $c_1',\ldots,c_k'$ of $\mathcal{K}$.
There must be a server $c_i$ for which $c_i'=r'$.
If the server matched to $c_i'$ can reach $r'$ in this turn, move all servers towards their counterparts in the matching with the maximum possible speed of $(1+\delta)m_s$.
Otherwise, select the server $\tilde{a}$ which is closest to $r'$ and move it to $r'$ with speed at most $(1+\frac{\delta}{2})m_s$.
All other servers move towards their counterparts in the matching with speed $(1+\delta)m_s$.

\medskip

We briefly want to discuss the fact that both steps of our algorithm are necessary for a bounded competitiveness.
For the classical $k$-Server Problem, a simple greedy algorithm, which always moves the closest server onto the request has an unbounded competitive ratio.
We can show, that a simple algorithm which just tries to imitate any $k$-Server algorithm as best as possible is also not successful.
Intuitively, the simulated algorithm can move many servers towards the request within one time step and serve the following sequence with them, while the online algorithm needs
multiple time steps to get the corresponding servers in position due to the speed limitation.
\smallskip

Simple algorithm:
Let $\mathcal{K}$ be any given $k$-Server algorithm.
The $k$-Mobile Server algorithm does the following: Simulate $\mathcal{K}$. Compute a minimum weight matching (with the distances as weights) between the own servers and the servers of $\mathcal{K}$.
Move every server towards the matched server at maximum speed.

\begin{theorem}
\label{th:simplealgo}
For $k\geq 2$, there are competitive $k$-Server algorithms such that the simple algorithm for the $k$-Mobile Server Problem does not achieve a competitive ratio independent of $n$.
\end{theorem}

The remainder of this section is devoted to the analysis of the competitive ratio of the UMS algorithm.
In Section~\ref{sec:trivialunweighted}, we first consider the case that the distance between consecutive requests $m_c$ is smaller than the movement speed of the algorithm's servers.
This case is easier than the case of slower servers since we can always guarantee that the online algorithm has one server on the position of the request.
In the other case ($m_c\geq(1+\delta)m_s$), described in Section~\ref{sec:projection}, we need to extend our analysis to incorporate situations in which our online algorithm has no server near the request although the optimal offline solution might have such a server.
Details left out due to space constraints can be found in Appendix~\ref{app:unweighted}.

%%%%%%%%%%%%%%%%%%%%%%%%%%%%%%%%%%%%%%%%%%%%%%%%%%%%%%%%%%%%%%%%%%%%%%%%%%%%%%%%%%%%%%%%%%%%%%%%%%%
%%%%%%%%%%%%%%%%%%%%%%%%%%%%%%%%%%%%%%%%%%%%%%%%%%%%%%%%%%%%%%%%%%%%%%%%%%%%%%%%%%%%%%%%%%%%%%%%%%%

\subsection{Fast Resource Movement}\label{sec:trivialunweighted}

We first deal with the case that $m_c\leq (1 -\varepsilon)\cdot m_s$ for some $\varepsilon\in (0,1)$. We show that we can achieve a result independent of the input length, even without resource augmentation. At the end of this section, we briefly discuss how to extend the result to incorporate resource augmentation, i.e., if the online algorithm has a maximum movement distance of $(1+\delta)m_s$, we handle all cases with $m_c\leq (1+\delta -\varepsilon)\cdot m_s$.

\begin{theorem}
\label{th:trivialunweighted}
If $m_c\leq (1 -\varepsilon)\cdot m_s$ for some $\varepsilon\in (0,1)$, the algorithm UMS is $\nicefrac{2}{\varepsilon}\cdot c(\mathcal{K})$-competitive, where $c(\mathcal{K})$ is the competitive ratio of the simulated $k$-server algorithm $\mathcal{K}$.
\end{theorem}

\begin{proof}
We assume the servers adapt their ordering $a_1,\ldots,a_k$ according to the minimum matching in each time step.
Based on the matching, we define the potential
$\psi:=\frac{2}{\varepsilon}\cdot\sum_{i=1}^{k}d(a_i,c_i)$.
Note that the algorithm reaches the point of $r$ in each time step, and hence only pays for the movement of its servers,
i.e., $C_{Alg}=\sum_{i=1}^{k}d(a_i,a_i')$.
We assume, that $c_1$ is on the request after the current time step, i.e., $c_1'=r'$.

First, consider the case that $a_1$ can reach $r'$ in this time step.
Since each server moves directly towards their counterpart in the matching, we have

$\begin{array}{rcl}
  \Delta\psi &=& \frac{2}{\varepsilon}\cdot\sum_{i=1}^{k}d(a_i',c_i') - \frac{2}{\varepsilon}\cdot\sum_{i=1}^{k}d(a_i,c_i) \\
  &\leq& \frac{2}{\varepsilon}\cdot\sum_{i=1}^{k} d(c_i,c_i') - \frac{2}{\varepsilon}\cdot\sum_{i=1}^{k} d(a_i,a_i') \\
  &=& \frac{2}{\varepsilon}\cdot C_\mathcal{K} - \frac{2}{\varepsilon}\cdot C_{Alg}.
\end{array}$

Now assume that $a_1$ cannot reach $r'$ in this time step.
The server moves at full speed and hence $d(a_1',c_1')-d(a_1,c_1')=-m_s$.
Now, let $a_2$ be the server which is at range at most $m_c$ to $r'$ and does the greedy move possibly away from $c_2'$ onto $r'$.
It holds $d(a_2',c_2')-d(a_2,c_2')\leq m_c$.
In total, we get

$\begin{array}{rcl}
  \Delta\psi &\leq& \frac{2}{\varepsilon}(\sum_{i=1}^{k}d(a_i',c_i') - \sum_{i=1}^{k}d(a_i,c_i')) + \frac{2}{\varepsilon}\sum_{i=1}^{k}d(c_i,c_i') \\
  &\leq& \frac{2}{\varepsilon}(d(a_1',c_1')-d(a_1,c_1') + d(a_2',c_2')-d(a_2,c_2')) - \frac{2}{\varepsilon}\sum_{i=3}^{k}d(a_i,a_i') + \frac{2}{\varepsilon}\sum_{i=1}^{k}d(c_i,c_i') \\
  &\leq& -2m_s - \frac{2}{\varepsilon}\sum_{i=3}^{k}d(a_i,a_i') + \frac{2}{\varepsilon}\cdot C_\mathcal{K} \\
  &\leq& - \sum_{i=1}^{k}d(a_i,a_i') + \frac{2}{\varepsilon}\cdot C_\mathcal{K}.
\end{array}$

\end{proof}

We can extend this bound to the resource augmentation scenario, where the online algorithm may move the servers a maximum distance of $(1+\delta)\cdot m_s$.
When relaxing the condition appropriately to $m_c\leq (1+\delta-\varepsilon)\cdot m_s$, we get the following result:

\begin{corollary}
If $m_c\leq (1+\delta -\varepsilon)\cdot m_s$ for some $\varepsilon\in (0,1)$, the algorithm UMS is $\frac{2\cdot(1+\delta)}{\varepsilon}\cdot c(\mathcal{K})$-competitive, where $c(\mathcal{K})$ is the competitive ratio of the simulated $k$-server algorithm $\mathcal{K}$.
\end{corollary}

The proof works the same as above by replacing occurrences of $m_s$ by $(1+\delta)m_s$ and changing the potential to $\frac{2\cdot(1+\delta)}{\varepsilon}\sum_{i=1}^{k}d(a_i,c_i)$.

At first glance, the result seems to become weaker with increasing $\delta$ if $\varepsilon$ stays the same.
The reason is that by fixing $\varepsilon$ the relative difference $((1+\delta)m_s-m_c)/m_s$ between $m_c$ and $(1+\delta)m_s$ actually decreases,
i.e., relatively speaking, $m_c$ gets closer to $(1+\delta)m_s$.
It can be seen that if instead we fix the value of $m_c$ and increase $\delta$, the value of $\varepsilon$ increases by the same amount and hence the competitive ratio tends towards $2\cdot c(\mathcal{K})$.

%%%%%%%%%%%%%%%%%%%%%%%%%%%%%%%%%%%%%%%%%%%%%%%%%%%%%%%%%%%%%%%%%%%%%%%%%%%%%%%%%%%%%%%%%%%%%%%%%%%
%%%%%%%%%%%%%%%%%%%%%%%%%%%%%%%%%%%%%%%%%%%%%%%%%%%%%%%%%%%%%%%%%%%%%%%%%%%%%%%%%%%%%%%%%%%%%%%%%%%
\subsection{Slow Resource Movement}

This section considers the case $m_c\geq(1+\delta)m_s$ and is structured as follows:
To support our potential argument, we first introduce a transformation of the simulated $k$-Server algorithm which ensures that the simulated servers are always located near the request.
We then introduce an abstraction of the offline solution, reducing it to the positioning of a single server $\hat{o}$ which acts as a reference point for a new potential function.
The server $\hat{o}$ approximates the optimal positioning of the servers while at the same time obeys certain movement restrictions necessary in our analysis.
Finally, we complete the analysis by combining the new derived potential function with the methods from the previous section.

\paragraph{The k-Server Projection}
\label{sec:projection}

Our goal is to transform a $k$-Server algorithm $\mathcal{K}$ into a $k$-Server algorithm $\hat{\mathcal{K}}$ which serves the requests of a $k$-Mobile Server instance such that all servers keep relatively close to the current request $r$.
For the case $m_c\geq(1+\delta)m_s$, we want our algorithm to use this projection as a simulated algorithm as opposed to a regular $k$-Server algorithm, hence we must ensure that this projection is computable online with the information available to our online algorithm.
The servers of $\mathcal{K}$ are denoted as $c_1,\ldots, c_k$ and the servers of $\hat{\mathcal{K}}$ as $\hat{c}_1,\ldots, \hat{c}_k$.

We define two circles around $r$: The inner circle $inner(r)$ has a radius of $4 k\cdot m_c$ and the outer circle $outer(r)$ has a radius of $(8 k + 1)\cdot m_c$.
We will maintain $\hat{c}_i\in outer(r)$ for the entirety of the execution.
The time is divided into phases,
where the phase starting at time $t$ with the request at point $r_t$ ends on the smallest $t'>t$ such that $d(r_t,r_t')\geq 4k\cdot m_c$.
During a phase the simulated servers move to preserve the following:
If $c_i\in inner(r)$, then $\hat{c}_i=c_i$.
At the end of the phase, in addition to the previous condition, it should hold:
If $c_i\notin inner(r)$, then $\hat{c}_i$ is on the boundary of $inner(r)$ such that $d(c_i,\hat{c}_i)$ is minimized.
It is obvious that the definition of the algorithm guarantees $\hat{c}_i\in outer(r)$ for all $i$ at each point in time.

\begin{proposition}
\label{prop:projectionunweighted}
For the servers $\hat{c}_1,\ldots, \hat{c}_k$ of $\hat{\mathcal{K}}$ it holds $d(\hat{c}_i,r)\leq(8 k + 1)\cdot m_c$ during the whole execution.
The costs of $\hat{\mathcal{K}}$ are at most $\mathcal{O}(k)$ times the costs of $\mathcal{K}$.
\end{proposition}

%%%%%%%%%%%%%%%%%%%%%%%%%%%%%%%%%%%%%%%%%%%%%%%%%%%%%%%%%%%%%%%%%%%%%%%%%%%%%%%%%%%%%%%%%%%%%%%%%%%
%%%%%%%%%%%%%%%%%%%%%%%%%%%%%%%%%%%%%%%%%%%%%%%%%%%%%%%%%%%%%%%%%%%%%%%%%%%%%%%%%%%%%%%%%%%%%%%%%%%

\paragraph{The Offline Helper}
\label{sec:helper}

We define a new offline server $\hat{o}$, which approximates the optimal position $o^{*}$ while managing the role change of $o^{*}$ in a smooth manner.
By $\hat{a}$, we denote the server of the online algorithm with minimal distance to $\hat{o}.$
For a formal description of the behavior, we need the following definitions:
\begin{itemize}
  \item The inner circle $inner_t(o_i)$ contains all points $p$ with $d_t(o_i,p)\leq \frac{\delta^{2}}{48960k}\cdot d_t(o_i,o_i^{a})$.
	\item The outer circle $outer_t(o_i)$ contains all points $p$ with $d_t(o_i,p)\leq \frac{\delta}{48}\cdot d_t(o_i,o_i^{a})$.
\end{itemize}
Abusing notation, we also refer to $inner_t(o_i)$ and $outer_t(o_i)$ as distances equal to the radius defined above.
This section is devoted to proving the following:

\begin{proposition}
\label{th:invariant}
There exists a virtual server $\hat{o}$ which moves at a speed of at most $(2+\frac{1020k}{\delta})\cdot m_c$ per time step, for which $d(\hat{a},\hat{o}) \leq 2\cdot d(o^{*},o^{*a}) + d(a^{*},r)$ at all times, and for which the following conditions hold as long as $d_t(o^{*},o^{*a}) \geq 2\cdot 51483\frac{km_c}{\delta^2}$:
\begin{enumerate}
  \item If $r\in inner(o^{*})$ at the end of the current time step, $\hat{o}$ moves at a maximum speed of $(1+\frac{\delta}{8})m_s$,
	  i.e., $r_t\in inner_t(o^{*}) \Rightarrow d(\hat{o}^{(t-1)},\ \hat{o}^{(t)})\leq (1+\frac{\delta}{8})m_s$.
	\item If $r\in inner(o^{*})$ at the end of the current time step, then $\hat{o}\in outer(o^{*})$ at the end of the current time step,
	  i.e., $r_t\in inner_t(o^{*}) \Rightarrow \hat{o}^{(t)}\in outer_t(o^{*})$.
\end{enumerate}
\end{proposition}

In the following, we show that it is possible to define a movement pattern for $\hat{o}$ in a way, such that invariants 1 and 2 of Proposition~\ref{th:invariant} hold as long as $d(o^{*},o^{*a}) \geq 51483\frac{km_c}{\delta^2}$.
Otherwise, $\hat{o}$ will simply follow $r$ and restore the properties once $d(o^{*},o^{*a}) \geq 2\cdot 51483\frac{km_c}{\delta^2}$.
In order to describe the movement in detail, we introduce the concept of transitions.

In the input sequence and a given optimal solution, we define a \emph{transition} between two steps $t_1<t_2$, if there are $o_i,\ o_j$ such that $o_i=o^{*}$ and $r\in inner_{t_1}(o_i)$ at time step $t_1$ and $o_j=o^{*}$ and $r\in inner_{t_2}(o_j)$ at time step $t_2$.
In between these two time steps, $r\notin inner(o^{*})$.
For such a transition, we define the transition time as $t^{*}:=t_2-t_1$.
If $t^{*}>inner_{t_1}(o^{*})/m_c + 2$, we call this a \emph{long transition}.
Otherwise, we call it a \emph{short transition}.
We say that $o_i$ \emph{passes} the request after $t_1$ and $o_j$ \emph{receives} the request in $t_2$.
The concept is illustrated in Figure~\ref{figure:Transition}.

\begin{figure}[ht]
	\begin{minipage}[c]{0.57\textwidth}
		\includegraphics[page=1, width=\textwidth, trim = 8.25cm 6.25cm 8.75cm 1cm, clip=true]{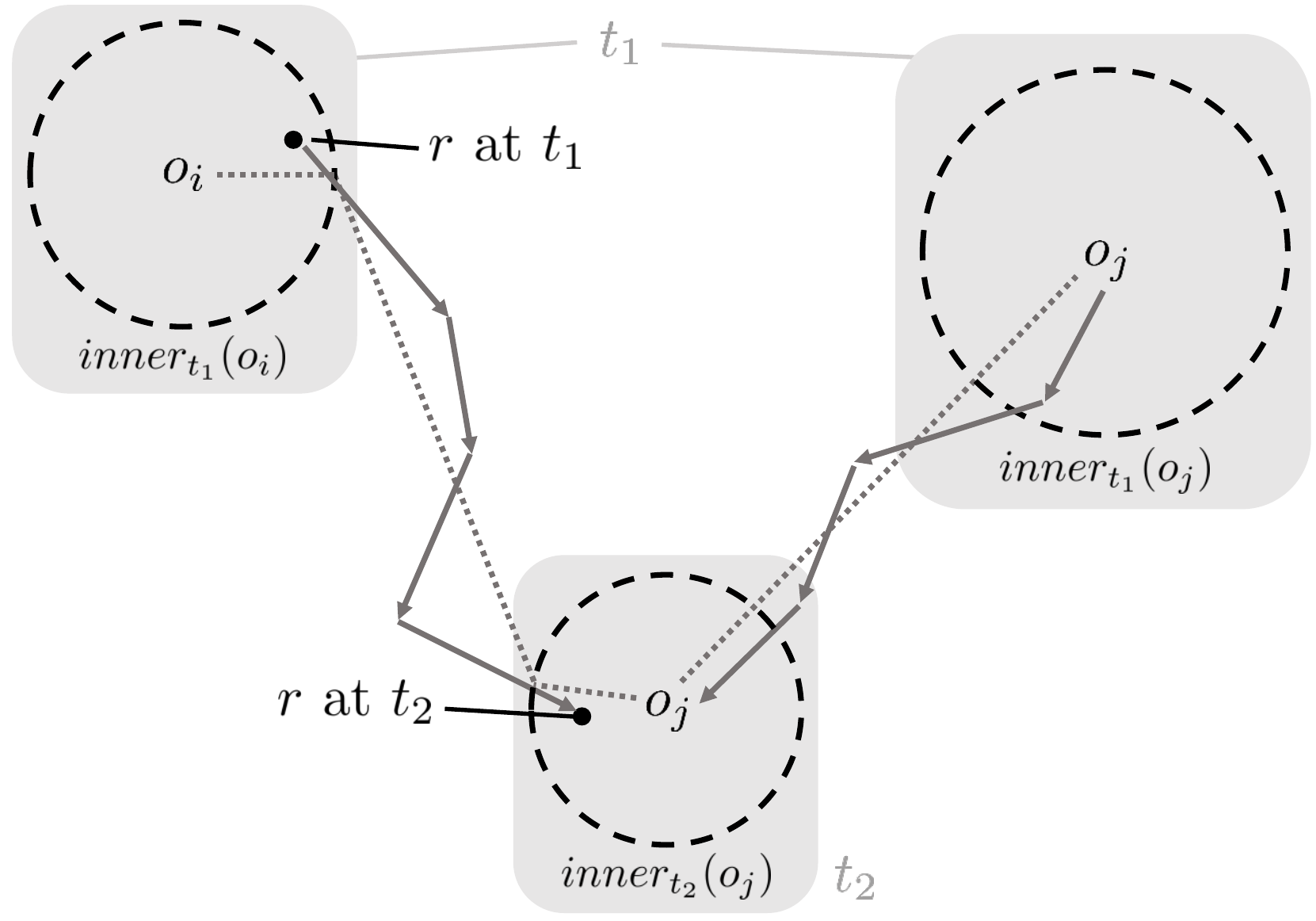}
	\end{minipage}
	\begin{minipage}{0.04\textwidth}
		\hfill
	\end{minipage}
	\begin{minipage}[c]{0.37\textwidth}
		\captionof{figure}{Example for a transition from $o_i$ to $o_j$.
		By definition, \( r \) crosses the border of \( inner(o_{i}) \) after time step \( t_{1} \) ($o_i$ passes $r$ after $t_1$).
		The transition stops at step \( t_{2} \) when \( r \) has entered \( inner_{t_{2}}(o_{j}) \) ($o_j$ receives $r$ in $t_2$).
		Note that \( o_{j} \)'s position and the radius of its inner circle may change from \( t_{1} \) to \( t_{2} \).
		The distance moved by \( r \) is at most \(  (t_2-t_1)\cdot m_{c} \).
		The dotted line represents the estimation of $d_{t_1}(o_i,o_j)$ used in Lemma~\ref{le:shorttrans}.}
		\label{figure:Transition}
	\end{minipage}
\end{figure}

\noindent
The behavior of $\hat{o}$ can be computed as follows:
\begin{enumerate}
  \item During a long transition between time steps $t_1$ and $t_2$, move with speed $d(\hat{o}^{(t-1)},\ \hat{o}^{(t)})\leq(2+\frac{1020k}{\delta})\cdot m_c$ towards $r_t$ whenever $r_t\notin inner_t(o^{*})$.
	  In the last two steps $t_2-1$ and $t_2$, move such that $\hat{o}^{(t_2-1)}=r_{t_2}$ at time $t_2-1$ and do not move in $t_2$ at all.
	  Informally, $\hat{o}$ moves one step ahead of $r$ such that $\hat{o}=r$ after the transition, as soon as $r\in inner(o^{*})$.
	\item For a sequence of short transitions starting with $o^{*}=o_i$ in $t_1$, determine which of the following events terminating the current sequence occurs first:
	  \begin{enumerate}
		  \item A long transition from a server $o_\ell$ to $o_j$ between time $t_2$ and $t_3$ occurs.
			  In this case, $\hat{o}$ simply moves towards $o_\ell^{(t)}$ in each step $t$ with speed at most $(1+\frac{\delta}{8})m_s$ until $t_2$.
				
			\item A short transition from a server $o_\ell$ to $o_j$ between time $t_2$ and $t_3$ occurs, where at one point prior in the sequence $d(o_j,o^{*}) > outer(o^{*})/3$.
			  If $\hat{o}$ can move straight towards the final position of $o_j$ in $t_3$ with speed $(1+\frac{\delta}{8})m_s$ without ever leaving $outer(o^{*})$, then do that.
				Otherwise move towards a point $p$ with $d(p,o_\ell)=\frac{2\delta}{145}\cdot d(o_\ell,o_\ell^{a})$.
				Among those candidates, $p$ minimizes $d(p,o_\ell^{(t_3)})$.
				When this point is reached, keep the invariant $d(\hat{o},o_\ell)=\frac{2\delta}{145}\cdot d(o_\ell,o_\ell^{a})$ whenever the final position of $o_j$ is not within $\frac{2\delta}{145}\cdot d(o_\ell,o_\ell^{a})$ around $o_\ell$.
				The position of $\hat{o}$ on the circle around $o_\ell$ should be the one closest to $o_j$'s final position.
				When $o_j^{(t_3)}$ is inside the circle, the position of $\hat{o}$ should be equal to $o_j^{(t_3)}$.

		\end{enumerate}
		\item If $d_{t_1}(o^{*},o^{*a}) < 51483\frac{km_c}{\delta^2}$, treat the time until $d_{t_2}(o^{*},o^{*a}) \geq 2\cdot 51483\frac{km_c}{\delta^2}$
		  as a long transition between $t_1$ and $t_2$, i.e., move towards $r$ with speed $(2+\frac{1020k}{\delta})\cdot m_c$ and skip one step ahead of $r$ during the last 2 time steps.
			(Steps 1 and 2 are not executed during this time.)
\end{enumerate}

Note that the server $\hat{o}$ is a purely analytical tool and hence the behavior as described above does not have to be computable online.

Our goal is to show that all invariants described in Proposition~\ref{th:invariant} hold inductively over all transitions.
We divide the entire timeline into sequences, where each sequence starts with both $r$ and $\hat{o}$ being in $inner(o^{*})$.
A sequence ends when one of the events stated in step 2 of the algorithm completes.
The following lemma states that the initial condition is restored after every long transition. 

\begin{lemma}
\label{le:longtrans}
If $\hat{o}\in outer_{t_1}(o^{*})$ at the beginning of a long transition between $t_1$ and $t_2$, then $\hat{o}\in inner_{t_2}(o^{*})$ at the end of the transition.
\end{lemma}

Our next goal is to analyze a sequence of short transitions.
During these transitions, $r$ moves faster than $\hat{o}$ and hence the distance of $\hat{o}$ to $o^{*}$ increases due to the role change after a transition.
The next lemma establishes an upper bound on that increase.
Since we use the lemma in another context as well, the formulation is slightly more general.

\begin{lemma}
\label{le:shorttrans}
Every short transition between $o_i$ in step $t_1$ and $o_j$ in step $t_2$ can increase the distance of some server $s$, which moves at speed at most $(1+\delta)m_s$, to $o^{*}$ by at most
$\min\{6.001\cdot\frac{\delta^{2}}{48960k}\cdot d_{t_1}(o^{*},o^{*a}) + 8.001m_c \ ,\ 6.002\cdot \frac{\delta^{2}}{48960k}\cdot d_{t_2}(o^{*},o^{*a}) + 8.002m_c\}$.

Likewise, $s$ decreases its distance to $o^{*}$ by at most $\min\{6.001\cdot\frac{\delta^{2}}{48960k}\cdot d_{t_1}(o^{*},o^{*a}) + 8.001m_c \ ,\ 6.002\cdot \frac{\delta^{2}}{48960k}\cdot d_{t_2}(o^{*},o^{*a}) + 8.002m_c\}$.
\end{lemma}

We want to show, that $\hat{o}\in inner(o^{*})$ holds after a sequence of short transitions is terminated by one of the conditions described in step 2 of the algorithm.
During the sequence, we must also show that $\hat{o}\in outer(o^{*})$.
The main idea for the following lemma is that $o_\ell$ never leaves $outer(o^{*})/3$ per definition and hence following it keeps $\hat{o}$ inside $outer(o^{*})$.

\begin{lemma}
\label{le:longtermination}
Consider a sequence of short transitions which is terminated by a long transition.
If $\hat{o}\in inner(o^{*})$ at the beginning of the sequence, then $\hat{o}\in inner(o^{*})$ after the long transition.
During the sequence of short transitions, $\hat{o}\in outer(o^{*})$.
\end{lemma}

We show with the help of Lemma~\ref{le:shorttrans} that during the sequence of transitions, $\hat{o}$ does not loose too much distance to $o^{*}$, while
$o_j$, since at one point $d(o_j,o^{*}) > outer(o^{*})/3$, takes enough time to get into position for a short transition such that $\hat{o}$ can reach the final position
of $o_j$ in time.

\begin{lemma}
\label{le:shorttermination}
Consider a sequence of short transitions which is terminated by a short transition from $o_\ell$ to $o_j$, where at one point prior in the sequence $d(o_j,o^{*}) > outer(o^{*})/3$.
If $\hat{o}\in inner(o^{*})$ at the beginning of the sequence and $d(o^{*},o^{*a}) \geq 51483\frac{km_c}{\delta^2}$ at all times, then $\hat{o}\in inner(o^{*})$ after the transition to $o_j$.
During the sequence, $\hat{o}\in outer(o^{*})$.
\end{lemma}

Our analysis of the movement pattern of $\hat{o}$ leads directly to the following lemma, in which we mostly need to argue that either $\hat{o}\in outer(o^{*})$ or $\hat{o}=r$.

\begin{lemma}
\label{le:helperdistance}
During the execution of the algorithm, $d(\hat{a},\hat{o}) \leq 2\cdot d(o^{*},o^{*a}) + d(a^{*},r)$ as long as the algorithm is in step 1 or 2.
\end{lemma}

So far we have shown that all claims of Proposition~\ref{th:invariant} hold as long as the algorithm is not in step 3.
It remains to analyze step 3 of the algorithm, using similar arguments as for analyzing the long transitions earlier.

\begin{lemma}
\label{le:stepthree}
After the execution of step 3 it holds $\hat{o}=r$.
Furthermore, $d(\hat{a},\hat{o}) \leq 2\cdot d(o^{*},o^{*a}) + d(a^{*},r)$ during step 3 of the algorithm.
\end{lemma}

%%%%%%%%%%%%%%%%%%%%%%%%%%%%%%%%%%%%%%%%%%%%%%%%%%%%%%%%%%%%%%%%%%%%%%%%%%%%%%%%%%%%%%%%%%%%%%%%%%%
%%%%%%%%%%%%%%%%%%%%%%%%%%%%%%%%%%%%%%%%%%%%%%%%%%%%%%%%%%%%%%%%%%%%%%%%%%%%%%%%%%%%%%%%%%%%%%%%%%%

\paragraph{Algorithm Analysis}
\label{sec:analysis}

We now turn our attention back to the analysis of the UMS algorithm.
In the following, we assume $\mathcal{K}$ to be a $k$-Server algorithm obtained from \cref{prop:projectionunweighted}.
We use a potential composed of two major parts which balance the main ideas of our algorithm against each other:
$\phi$ will measure the costs of the greedy strategy, while $\psi$ will cover the matching to the simulated $k$-Server algorithm.

Let $\hat{o}$ be an offline server which fulfills the invariants stated in Proposition~\ref{th:invariant}.
Recall that $\hat{a}$ denotes the currently closest server of the online algorithm to $\hat{o}$.
The first part of the potential is then defined as
$$\phi:=\left\{\begin{array}{ll}
  4\cdot d(\hat{a},\hat{o}) & \text{ if } d(\hat{a},\hat{o})\leq 107548\cdot \frac{k m_c}{\delta^2}\\
	4\cdot \frac{1}{\delta m_s}d(\hat{a},\hat{o})^2 - A & \text{ if } 107548\cdot \frac{k m_c}{\delta^2}<d(\hat{a},\hat{o})
\end{array}\right.$$
with $A:=4\cdot(\frac{1}{\delta m_s}(107548\frac{k m_c}{\delta^2})^2 - 107548\frac{k m_c}{\delta^2})$.

For the second part, we set
$$\psi:=Y\cdot\frac{m_c}{\delta m_s}\sum\limits_{i=1}^{k}d(a_i,c_i)$$
where the online servers $a_i$ are always sorted such that they represent a minimum weight matching to the simulated servers $c_i$.
We choose $Y=\Theta(\frac{k}{\delta^2})$ to be sufficiently large.

If we understand $\phi$ as a function in $d(\hat{a},\hat{o})$, then we can rewrite it as $\phi(d(\hat{a},\hat{o})) = \max\{4\cdot d(\hat{a},\hat{o}), 4\cdot\frac{1}{\delta m_s}d(\hat{a},\hat{o})^2 - A\}$.
Hence, when estimating the potential difference $\Delta\phi=\phi(d(\hat{a}',\hat{o}'))-\phi(d(\hat{a},\hat{o}))$, we can upper bound it by replacing the term $\phi(d(\hat{a},\hat{o}))$ with the case identical to $\phi(d(\hat{a}',\hat{o}'))$.
This mostly reduces estimating $\Delta\phi$ to bounding the difference $d(\hat{a}',\hat{o}')-d(\hat{a},\hat{o})$.

For some of our estimations we use a slightly altered result from~\cite{DBLP:conf/spaa/FeldkordH17}.

\begin{lemma}
\label{le:Geo}
Let $s$ be some server with
$d(s',r')\leq\frac{\sqrt\delta}{2}\cdot d(a_i',r')$ and $a_i$ moves towards $r'$ a distance of $d(a_i,a_i')$, then $d(a_i,s')-d(a_i',s')\geq\frac{1+\frac{1}{4}\delta}{1+\frac{1}{2}\delta}d(a_i,a_i')$.
\end{lemma}

We start the analysis by bounding the second potential difference $\Delta\psi$.
The bounds can be obtained by similar arguments as in the proof of Theorem~\ref{th:trivialunweighted}.

\begin{lemma}
\label{le:easycancel}
$\Delta\psi\leq Y\cdot\frac{m_c}{\delta m_s}\cdot C_{\mathcal{K}} - \sum_{i=1}^{k}d(a_i,a_i')$.
\end{lemma}

\begin{lemma}
\label{le:massivecancel}
  If $d(a^{*'},r')>0$, then $\Delta\psi\leq Y\cdot\frac{m_c}{\delta m_s}C_\mathcal{K} - \sum\limits_{i=1}^{k}d(a_i,a_i') - \frac{Y-4}{2}m_c.$
\end{lemma}

Now consider the case that $r'\notin inner(o^{*'})$.
We have $d(a^{*'},r')\leq d(o^{*a'},r') \leq d(o^{*'},o^{*a'}) + d(o^{*'},r') \leq (\frac{48960k}{\delta^{2}}+1)\cdot d(o^{*'},r')$.
The movement cost are canceled by $\Delta\psi$ as in Lemma~\ref{le:easycancel}.
It only remains to bound the possible increase of $\phi$.
We use $d(\hat{a}',\hat{o}') - d(\hat{a},\hat{o}) \leq (3+\frac{1020k}{\delta})\cdot m_c$.

\begin{enumerate}
  \item $d(\hat{a}',\hat{o}')\leq 107548\cdot \frac{k m_c}{\delta^2}$:
	  $\Delta\phi\leq 4\cdot d(\hat{a}',\hat{o}')\leq 8\cdot d(o^{*'},o^{*a'}) + 4\cdot d(a^{*'},r') \leq (12\cdot\frac{48960k}{\delta^{2}} + 4)\cdot d(o^{*'},r')$.
	
	\item $107548\cdot \frac{k m_c}{\delta^2}<d(\hat{a}',\hat{o}')$:
	  $\Delta\phi\leq \frac{4}{\delta m_s}(d(\hat{a}',\hat{o}')^2-d(\hat{a},\hat{o})^2) \leq \frac{4}{\delta m_s}(d(\hat{a}',\hat{o}')^2-(d(\hat{a}',\hat{o}')-(3+\frac{1020k}{\delta})\cdot m_c)^2) \leq \mathcal{O}(\frac{k}{\delta})\cdot\frac{m_c}{\delta m_s}d(\hat{a}',\hat{o}')\leq \mathcal{O}(\frac{k^2}{\delta^{3}})\cdot \frac{m_c}{\delta m_s}d(o^{*'},r')$.
\end{enumerate}
In all of the above, the competitive ratio is bounded by $\mathcal{O}(\frac{k^2}{\delta^{3}})\cdot \frac{m_c}{\delta m_s} + Y\cdot\frac{m_c}{\delta m_s}\cdot c(\mathcal{K})$.

Finally, we consider the case $r'\in inner(o^{*'})$.
Whenever $d(a^{*},r')> 102970\frac{k m_c}{\delta^2}$, we use \cref{le:Geo} to obtain the following:

\begin{lemma}
\label{le:phicancel}
If $d(a^{*},r')> 102970\frac{k m_c}{\delta^2}$ and $r'\in inner(o^{*'})$, then $d(a_i',\hat{o}')-d(a_i,\hat{o})\leq -\frac{\delta}{8}m_s$.
\end{lemma}

With this lemma, $\phi$ can be used to cancel the costs of the algorithm in case of a high distance to $r$.

\begin{lemma}
\label{le:unweightedfinal}
If $r'\in inner(o^{*'})$, then $C_{Alg} + \Delta\phi + \Delta\psi \leq Y\cdot\frac{m_c}{\delta m_s}\cdot C_{\mathcal{K}} + 2\cdot d(o^{*'},r')$.
\end{lemma}

The resulting competitive ratio of $Y\cdot\frac{m_c}{\delta m_s}\cdot c(\mathcal{K}) + 2$
is less than the $\mathcal{O}(\frac{k^2}{\delta^{3}})\cdot \frac{m_c}{\delta m_s} + Y\cdot\frac{m_c}{\delta m_s}\cdot c(\mathcal{K})$ bound from the former set of cases.
Accounting for the loss due to the transformation of the simulated $k$-Server algorithm,
we obtain the following result:

\begin{theorem}
\label{th:unweighted}
If $m_c\geq\cdot (1+\delta)m_s$, the algorithm UMS is $\mathcal{O}(\frac{1}{\delta^4}\cdot k^2\cdot \frac{m_c}{ m_s} +\frac{1}{\delta^3}\cdot k^2 \cdot\frac{m_c}{ m_s}\cdot c(\mathcal{K}))$-competitive, where $c(\mathcal{K})$ is the competitive ratio of the simulated $k$-server algorithm $\mathcal{K}$.
\end{theorem}

%%%%%%%%%%%%%%%%%%%%%%%%%%%%%%%%%%%%%%%%%%%%%%%%%%%%%%%%%%%%%%%%%%%%%%%%%%%%%%%%%%%%%%%%%%%%%%%%%%%
%%%%%%%%%%%%%%%%%%%%%%%%%%%%%%%%%%%%%%%%%%%%%%%%%%%%%%%%%%%%%%%%%%%%%%%%%%%%%%%%%%%%%%%%%%%%%%%%%%%

\section{Extension to the Weighted Problem}
\label{sec:weighted}

In this section we consider our general model in which the movement costs are weighted with a factor $D>1$.
We assume throughout the section that $D\geq 2$ for convenience in the analysis.
In case $D<2$, we may just apply the algorithm from the previous section, whose costs increase by at most a factor of $2$ as a result.

The main difference to the unweighted case is that our algorithm uses a $k$-Page Migration algorithm as guidance,
whose best competitive ratio in the deterministic case so far is a factor $\Theta(k)$ worse than that of a $k$-Server algorithm for general metrics.
The analysis is slightly more involved since unlike in the $k$-Server Problem, a $k$-Page Migration algorithm does not always have to have one page at the point of the request.
In case of small distances to $r$, the movement costs have to be balanced against the serving costs by scaling down the movement distance by a factor of $D$.
Throughout this section, we use the same notation as for the unweighted version.

\medskip

\noindent
Our algorithm \emph{Weighted-Mobile Servers (WMS)} works as follows:\\
Take any $k$-Page Migration algorithm $\mathcal{K}$.
Upon receiving the next request $r'$, simulate the next step of $\mathcal{K}$.
Calculate a minimum weight matching (with the distances as weights) between the servers $a_1,\ldots,a_k$ of the online algorithm and the pages $c_1',\ldots,c_k'$ of $\mathcal{K}$.
Select the closest server $\tilde{a}$ to $r'$ and move it to $r'$ at most a distance $\min(m_c,\frac{1}{D}(1-\varepsilon)\cdot d(\tilde{a},r'))$ in case $m_c\leq(1+\delta-\varepsilon)m_s$ and at most $\min((1+\frac{\delta}{2})m_s,\frac{1}{D}(1-\frac{\delta}{2})\cdot d(\tilde{a},r'))$ in case $m_c\geq(1+\delta)m_s$. All other servers $a_i$ move towards their counterparts in the matching $c_i'$ with speed $\min((1+\delta)m_s,\frac{1}{D}\cdot d(\tilde{a},r'))$.
If another server than $\tilde{a}$ is closer to $r'$ after movement, then move all servers towards their counterpart in the matching with speed $m_s$ instead.

\medskip

The remainder of this section is devoted to the analysis of the WMS algorithm and is structured similar to Section~\ref{sec:unweighted}.
Due to space constraints, we can only give a brief overview and refer for the details to Appendix~\ref{app:weighted}.

We start by analyzing the case that $m_c\leq (1 -\varepsilon)\cdot m_s$ for some $\varepsilon \in (0,\frac{1}{2}]$.
For $\varepsilon\geq\frac{1}{2}$, our algorithm simply assumes $\varepsilon=\frac{1}{2}$.
It can be easily verified, that this does not hinder the analysis.

\begin{theorem}
\label{th:trivialweighted}
If $m_c\leq (1 -\varepsilon)\cdot m_s$ for some $\varepsilon\in (0,\frac{1}{2}]$, the algorithm WMS is $\nicefrac{\sqrt{2}\cdot 11}{\varepsilon}\cdot c(\mathcal{K})$-competitive, where $c(\mathcal{K})$ is the competitive ratio of the simulated $k$-Page Migration algorithm $\mathcal{K}$.
\end{theorem}

We can extend this bound to the resource augmentation scenario, where the online algorithm may move the servers a maximum distance of $(1+\delta)\cdot m_s$.
When relaxing the condition appropriately to $m_c\leq (1+\delta-\varepsilon)\cdot m_s$, then we get the following result:

\begin{corollary}
If $m_c\leq (1+\delta -\varepsilon)\cdot m_s$ for some $\varepsilon\in (0,\frac{1}{2}]$, the algorithm WMS is $\frac{\sqrt{2}\cdot 11\cdot(1+\delta)}{\varepsilon}\cdot c(\mathcal{K})$-competitive, where $c(\mathcal{K})$ is the competitive ratio of the simulated $k$-Page Migration algorithm $\mathcal{K}$.
\end{corollary}

%%%%%%%%%%%%%%%%%%%%%%%%%%%%%%%%%%%%%%%%%%%%%%%%%%%%%%%%%%%%%%%%%%%%%%%%%%%%%%%%%%%%%%%%%%%%%%%%%%%
%%%%%%%%%%%%%%%%%%%%%%%%%%%%%%%%%%%%%%%%%%%%%%%%%%%%%%%%%%%%%%%%%%%%%%%%%%%%%%%%%%%%%%%%%%%%%%%%%%%

Similar to the $k$-Server Projection discussed in Section~\ref{sec:projection}, we obtain the following result which gives us a new $k$-Page Migration algorithm needed for the case $m_c\geq (1+\delta)m_s$.

\begin{proposition}
\label{prop:projectionweighted}
Let $\mathcal{K}$ be an online algorithm for the $k$-Page Migration Problem.
There exists an online algorithm $\hat{\mathcal{K}}$ for the $k$-Page Migration Problem with pages $\hat{c}_1,\ldots, \hat{c}_k$ such that it holds $d(\hat{c}_i,r)\leq(32 k D + 1)\cdot m_c$ during the whole execution.
The costs of $\hat{\mathcal{K}}$ are at most $\mathcal{O}(k)$ times the costs of $\mathcal{K}$.
\end{proposition}

%%%%%%%%%%%%%%%%%%%%%%%%%%%%%%%%%%%%%%%%%%%%%%%%%%%%%%%%%%%%%%%%%%%%%%%%%%%%%%%%%%%%%%%%%%%%%%%%%%%
%%%%%%%%%%%%%%%%%%%%%%%%%%%%%%%%%%%%%%%%%%%%%%%%%%%%%%%%%%%%%%%%%%%%%%%%%%%%%%%%%%%%%%%%%%%%%%%%%%%

From here on we assume $\mathcal{K}$ to be a $k$-Page Migration algorithm obtained from the transformation in \cref{prop:projectionweighted}.
The offline helper and its invariants as stated in Proposition~\ref{th:invariant} do not depend on the simulated algorithm and therefore all insights gained from the corresponding section are still valid.
We use a potential composed of two major parts just as for the unweighted case.

Let $\hat{o}$ be an offline server which fulfills the invariants stated in Proposition~\ref{th:invariant}.
The first part of the potential is then defined as
$$\phi:=\left\{\begin{array}{ll}
  4\cdot d(\hat{a},\hat{o}) & \text{ if } d(\hat{a},\hat{o})\leq 107548D\cdot \frac{k m_c}{\delta^2}\\
	4\cdot \frac{1}{\delta m_s}d(\hat{a},\hat{o})^2 + A & \text{ if } 107548D\cdot \frac{k m_c}{\delta^2}<d(\hat{a},\hat{o})
\end{array}\right.$$
with $A:=4\cdot(107548D\frac{k m_c}{\delta^2} - \frac{1}{\delta m_s}(107548D\frac{k m_c}{\delta^2})^2)$.

For the second part, we set
$$\psi:=Y\cdot D\frac{m_c}{\delta m_s}\sum\limits_{i=1}^{k}d(a_i,c_i)$$
where the online servers $a_i$ are always sorted such that they represent a minimum weight matching to the simulated servers $c_i$.
We choose $Y=\Theta(\frac{k}{\delta^{2}})$ to be sufficiently large.

We begin by analyzing $\psi$, reusing ideas from the proof of Theorem~\ref{th:trivialweighted}.

\begin{lemma}
\label{le:easycancelW}
$\Delta\psi\leq \mathcal{O}(1)\cdot  Y\cdot\frac{m_c}{\delta m_s}\cdot C_{\mathcal{K}} - D\cdot\sum_{i=1}^{k}d(a_i,a_i')$.
\end{lemma}

\begin{lemma}
\label{le:massivecancelW}
If $d(a^{*'},r')>d(c^{*'},r')$, then $\Delta\psi\leq Y\cdot\frac{m_c}{\delta m_s}C_\mathcal{K} - D\cdot\sum\limits_{i=1}^{k}d(a_i,a_i') 
- \frac{Y-4}{2}D\frac{m_c}{\delta m_s}\cdot \min(m_s,\frac{1}{D}\cdot d(\tilde{a},r'))$.
\end{lemma}

Now consider the case that $r'\notin inner(o^{*'})$.
We have $d(a^{*'},r')\leq d(o^{*a'},r') \leq d(o^{*'},o^{*a'}) + d(o^{*'},r') \leq (\frac{48960k}{\delta^{2}}+1)\cdot d(o^{*'},r')$.
The movement costs are canceled by $\Delta\psi$ as in Lemma~\ref{le:easycancelW}.
It only remains to bound the possible increase of $\phi$.
We use $d(\hat{a}',\hat{o}') - d(\hat{a},\hat{o}) \leq (3+\frac{1020k}{\delta})\cdot m_c$.

\begin{enumerate}
  \item $d(\hat{a}',\hat{o}')\leq 107548D\cdot \frac{k m_c}{\delta^2}$:
	  $\Delta\phi\leq 4\cdot d(\hat{a}',\hat{o}')\leq 8\cdot d(o^{*'},o^{*a'}) + 4\cdot d(a^{*'},r') \leq (12\cdot\frac{48960k}{\delta^{2}} + 4)\cdot d(o^{*'},r')$.
	
	\item $107548D\cdot \frac{k m_c}{\delta^2}<d(\hat{a}',\hat{o}')$:
	  $\Delta\phi\leq \frac{4}{\delta m_s}(d(\hat{a}',\hat{o}')^2-d(\hat{a},\hat{o})^2) \leq \frac{4}{\delta m_s}(d(\hat{a}',\hat{o}')^2-(d(\hat{a}',\hat{o}')-(3+\frac{1020k}{\delta})\cdot m_c)^2) \leq \mathcal{O}(\frac{k}{\delta})\cdot\frac{m_c}{\delta m_s}d(\hat{a}',\hat{o}')\leq \mathcal{O}(\frac{k^2}{\delta^{3}})\cdot \frac{m_c}{\delta m_s}d(o^{*'},r')$.
\end{enumerate}
In all of the above, the competitive ratio is bounded by $\mathcal{O}(\frac{k^2}{\delta^{3}})\cdot \frac{m_c}{\delta m_s} + Y\cdot\frac{m_c}{\delta m_s}\cdot c(\mathcal{K})$.

Finally, we consider the case $r'\in inner(o^{*'})$.
As in the previous Section,
whenever $d(a^{*},r')> 102970D\frac{k m_c}{\delta^2}$, we make use of \cref{le:Geo} to obtain the following result, which then helps us bound $\Delta\phi$:

\begin{lemma}
\label{le:phicancelW}
If $d(a^{*},r')> 102970D\frac{k m_c}{\delta^2}$ and $r'\in inner(o^{*'})$, then $d(a_i',\hat{o}')-d(a_i,\hat{o}')\leq -\frac{\delta}{8}m_s$.
\end{lemma}

\begin{lemma}
\label{le:weightedfinal}
If $r'\in inner(o^{*'})$, then $C_{Alg} + \Delta\phi + \Delta\psi \leq Y\cdot\frac{m_c}{\delta m_s}\cdot C_{\mathcal{K}} + 2\cdot d(o^{*'},r')$.
\end{lemma}

The resulting competitive ratio $Y\cdot\frac{m_c}{\delta m_s}\cdot c(\mathcal{K}) + 2$
is less than the $\mathcal{O}(\frac{k^2}{\delta^{3}})\cdot \frac{m_c}{\delta m_s} + Y\cdot\frac{m_c}{\delta m_s}\cdot c(\mathcal{K})$ bound from the former set of cases.
Accounting for the loss due to the transformation of the simulated $k$-Page Migration algorithm,
we obtain the following upper bound:

\begin{theorem}
\label{th:weighted}
If $m_c\geq(1+\delta) m_s$, the algorithm WMS is $\mathcal{O}(\frac{1}{\delta^4}\cdot k^2\cdot \frac{m_c}{ m_s} +\frac{1}{\delta^3}\cdot k^2 \cdot\frac{m_c}{ m_s}\cdot c(\mathcal{K}))$-competitive, where $c(\mathcal{K})$ is the competitive ratio of the simulated $k$-Page Migration algorithm $\mathcal{K}$.
\end{theorem}

%%%%%%%%%%%%%%%%%%%%%%%%%%%%%%%%%%%%%%%%%%%%%%%%%%%%%%%%%%%%%%%%%%%%%%%%%%%%%%%%%%%%%%%%%%%%%%%%%%%
%%%%%%%%%%%%%%%%%%%%%%%%%%%%%%%%%%%%%%%%%%%%%%%%%%%%%%%%%%%%%%%%%%%%%%%%%%%%%%%%%%%%%%%%%%%%%%%%%%%

\section{Open Problems}

The gap between the upper and lower bound is closely related to the question of the deterministic upper bound for $k$-Page Migration:
Not only would an $\mathcal{O}(k)$-competitive algorithm for $k$-Page Migration directly improve the bound for $D>1$, it could also give an idea how to improve the analysis of the greedy step in our algorithm, such that the costly transformation of the simulated algorithm would no longer be needed.
This would potentially reduce the upper bound by another factor of $k$.
On the other hand, if $\Omega(k^2)$ is a lower bound for $k$-Page Migration, this carries over to our model as well.
We believe that the main algorithmic idea is suitable to reach an asymptotically optimal competitive ratio, but it remains an open problem to derive a proof of that.
The high constants in our proofs are partially due to allowing easier argumentation in certain segments of the proof.
There is however also great potential in reducing constants by trying to extend the potential analysis to operate in longer phases instead of doing a step-by-step analysis.

If we allow randomization, we can get an $\mathcal{O}(k)$-competitive $k$-Page Migration algorithm from~\cite{DBLP:journals/tcs/BartalCI01}.
As discussed in the related work section, the question of the best possible competitive ratio of randomized algorithms for the $k$-Server problem is still open,
however we know that a result polylogarithmic in $k$ can be achieved~\cite{DBLP:conf/focs/Lee18}.
As our construction is entirely deterministic, apart from potentially the simulated algorithm,
it would be interesting whether randomization can be used to significantly improve the competitive ratio.
The desired result would be an algorithm with a competitive ratio polylogarithmic in $k$.
More generally, the problem of finding a randomized algorithm with competitiveness $o(k)$ is still open for the classical $k$-Page Migration problem.

%%%%%%%%%%%%%%%%%%%%%%%%%%%%%%%%%%%%%%%%%%%%%%%%%%%%%%%%%%%%%%%%%%%%%%%%%%%%%%%%%%%%%%%%%%%%%%%%%%%
%%%%%%%%%%%%%%%%%%%%%%%%%%%%%%%%%%%%%%%%%%%%%%%%%%%%%%%%%%%%%%%%%%%%%%%%%%%%%%%%%%%%%%%%%%%%%%%%%%%

%%
%% Bibliography
%%

%% Either use bibtex (recommended), 

\bibliographystyle{plain}
\bibliography{references}

%% .. or use the thebibliography environment explicitly

%%%%%%%%%%%%%%%%%%%%%%%%%%%%%%%%%%%%%%%%%%%%%%%%%%%%%%%%%%%%%%%%%%%%%%%%%%%%%%%%%%%%%%%%%%%%%%%%%%%
%%%%%%%%%%%%%%%%%%%%%%%%%%%%%%%%%%%%%%%%%%%%%%%%%%%%%%%%%%%%%%%%%%%%%%%%%%%%%%%%%%%%%%%%%%%%%%%%%%%
\newpage
\appendix

\section{Details of Section~\ref{sec:lower-bounds}}
\label{app:lower}

\subsection{Proof of Theorem~\ref{th:lowerunbounded}}

\begin{figure}[H]
	\centering
	\includegraphics[page=5, width=0.7\textwidth, clip=true, trim = 9cm 9.45cm 9cm 5.5cm]{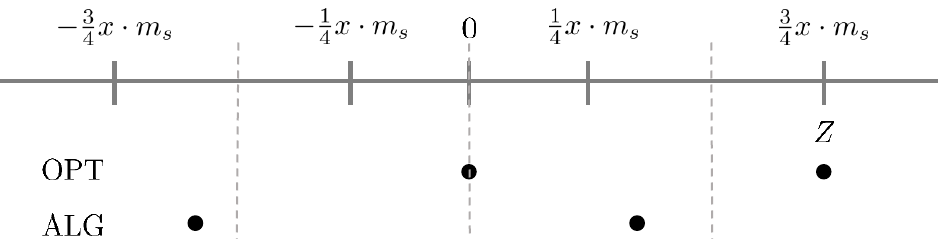}
	\caption{The line as used in the proof of \cref{th:lowerunbounded}. The circles indicate a possible configuration of the servers of the online algorithm and the optimal solution at the beginning of the second phase.
	The four segments are indicated by the dashed lines. The adversary has successfully chosen a segment which the online algorithm does not occupy.
	}
	\label{figure:line-proof-unbounded}
\end{figure}

{
	\renewcommand{\thetheorem}{\ref{th:lowerunbounded}}

\begin{theorem}
For $k\geq 2$, every randomized online algorithm for the $k$-Mobile Server Problem has a competitive ratio of at least $\Omega(\frac{n}{Dk^2})$,
where $n$ is the length of the input sequence.
\end{theorem}

}

\begin{proof}
We first explain the proof for $k=2$ in detail and then describe how to extend it to $k>2$ servers.
All servers start on the same position on the real line which we identify with 0.
The input proceeds in two phases.
In the first phase, there are $x\cdot m_s$ requests on point 0.
At the start of the second phase, choose one of the following points uniformly at random: $-\frac{3}{4}x\cdot m_s,-\frac{1}{4}x\cdot m_s,\frac{1}{4}x\cdot m_s,\frac{3}{4}x\cdot m_s$. We refer to this point as $Z$.
In the second phase, issue $\frac{x}{8}$ requests on $Z$.

The optimal solution moves one server to $Z$ during the first phase and has costs of at most $Dx\cdot m_s$.
Since the online algorithm only has two servers, both of its servers have a distance of at least $\frac{x}{8}\cdot m_s$ to $Z$ with probability at least $\frac{1}{2}$:
Divide the line into four segments of size $\frac{1}{4}x\cdot m_s$. The online algorithm can at most occupy two of theses segments (cf.\ Figure~\ref{figure:line-proof-unbounded}).
As a consequence, the expected costs for the online algorithm in the second phase are at least $\frac{1}{2}\cdot\sum_{i=1}^{\frac{x}{8}}(\frac{1}{8}x\cdot m_s-i\cdot m_s)\geq\frac{x^2}{264}m_s$.
The cost ratio is then $\Omega(x/D)=\Omega(n/D)$.

For $k>2$ servers, divide the line to the right of the starting point into $4(k-1)$ segments of size $x\cdot m_s$ each.
The segments are divided into $k-1$ groups of 4. Each group has two inner and two outer segments, where the outer segments neighbor segments of other groups.
The adversary now chooses in each group one of the two inner segments uniformly at random.
We refer to the middle point in each of the chosen segments as $Z_1,\ldots,Z_{k-1}$.
During the first phase, $kx$ requests appear at the starting point, and the adversary moves one server to $Z_1,\ldots,Z_{k-1}$ each, the last server remains at the starting point.
The moving costs for the adversary are $\mathcal{O}(Dk^2x\cdot m_s)$.

In the second phase, on each point $Z_1,\ldots,Z_{k-1}$, $\frac{x}{4}$ requests appear in order of distance to the starting point.
As before, if at the first time when a request appears on $Z_i$ and the online algorithm does not have one server in the corresponding segment, then the costs for serving requests for the online algorithm are at least $\Omega(x^2 m_s)$.
Now we iterate over the groups of segments: Consider the group which contains $Z_1$.
At the time of the first request on $Z_1$ the online algorithm either covers both, one or no inner segment of that group.
In case of only one covered segment, $Z_1$ lies in the other inner segment with probability $1/2$.
Consider a server in one of the inner segments: This server can not move into a neighboring group within $\nicefrac{x}{4}$ time steps.
Hence we can regard the servers which cover inner segments as "used up" for the following groups and hence we may apply the arguments inductively.
Let $a$, $b$ and $c$ the number of groups where the online algorithm covers both, one and no inner segment of that group respectively.
We have $a+b+c=k-1$, $2a+b\leq k$ and the expected number of segments for which the online algorithm incurs costs of $\Omega(x^2 m_s)$
are at least $a+\frac{1}{2}b$.
It is easy to see that the number of these segments are in $\Omega(k)$.

For the ratio we compare the costs and get $\frac{\Omega(k x^2 m_s)}{\mathcal{O}(Dk^2x\cdot m_s)}=\Omega(\frac{x}{Dk})=\Omega(\frac{n}{Dk^2})$.
\end{proof}

\subsection{Proof of Theorem~\ref{th:restrictedlower}}

{
	\renewcommand{\thetheorem}{\ref{th:restrictedlower}}
	
	\begin{theorem}
		For $k\geq 2$, every randomized online algorithm for the \( k \)-Mobile Server Problem, where the distance between consecutive requests is bounded by \( m_{c} \), has a competitive ratio of at least \( \Omega (\frac{m_{c}}{m_{s}}) \).
	\end{theorem}

}

\begin{proof}
	The proof follows the structure of the proof of \cref{th:lowerunbounded}.
	We first describe the bound for $k=2$ and then extend it to $k>2$ servers.
	All servers start on the same position denoted by \( 0 \) on the real line.
	The input is given in three phases.
	In the first phase, \( x \) requests are issued consecutively on point \( 0 \).
	At the beginning of the second phase, choose uniformly at random one of the points \(-\frac{3}{4}x\cdot m_s,-\frac{1}{4}x\cdot m_s,\frac{1}{4}x\cdot m_s,\frac{3}{4}x\cdot m_s\).
	Let the chosen point be \( Z \).
	Now, move the request by \( m_{c} \) towards \( Z \) in each time step until \( Z \) is reached.
	In the third phase, issue \( \frac{x}{8} \) requests consecutively at \( Z \).
	
	Observe that the request needs at most \( x\cdot\frac{m_s}{m_{c}} \) time steps to get to \( Z \).
	The optimal solution moves one of its servers in the first phase to \( Z \) and has a movement cost of at most $Dx\cdot m_{s}$.
	Since the distance between the request and an optimal server in each step in the second phase is at most $\frac{x}{2}\cdot m_{s}$,
	the costs of the optimal solution in this phase are bounded by $\frac{x^2}{2}\cdot \frac{m_s^2}{m_c}$.
	The optimal solution does not incur cost in the third phase.

	Since the point \( Z \) is unknown to the online algorithm, with a probability of at least \( \frac{1}{2} \), both servers of the online algorithm have a distance of at least \( \frac{1}{8} x \, m_{s} \) to $Z$ after the first phase.
	From here on, we assume that this is this the case.
	After the second phase the distance of an online server to the request is at least $\frac{1}{8}x m_s - x\cdot\frac{m_s^2}{m_c}=:y\cdot m_s$.
	The costs for serving requests in the third phase is minimized for the online algorithm, if it moves with speed $m_s$ towards $Z$ in each time step.
	The induced costs are at least $\sum_{i=0}^{y} (y-i)\cdot m_s\geq \frac{y^2}{2}m_s=\frac{1}{2}(\frac{1}{8}-\frac{m_s}{m_c})^2x^2 m_s\geq\Omega(x^{2}m_s)$ if $m_c$ is sufficiently large.
	
	In total, the competitive ratio is $\frac{\Omega(x^{2}m_s)}{Dxm_s + \nicefrac{x^2}{2}\cdot \nicefrac{m_s^2}{m_c}}=\Omega(\frac{m_{c}}{m_{s}})$ for sufficiently large $x$.
	
	Now consider the case of $k>2$ servers.
	We use a similar construction as in the proof of Theorem~\ref{th:lowerunbounded},
	but now divide the line to the right of the starting point into
	$5(k-1)$ segments of size $x\cdot m_s$ each.
The segments are divided into $k-1$ groups of 5. Each group has three inner and two outer segments, where the outer segments neighbor segments of other groups.
The adversary now chooses in each group one of the two inner segments, which neighbor an outer segment uniformly at random.
We refer to the middle point in each of the chosen segments as $Z_1,\ldots,Z_{k-1}$.
During the first phase, $kx$ requests appear at the starting point, and the adversary moves one server to $Z_1,\ldots,Z_{k-1}$ each, the last server remains at the starting point.
The moving costs for the adversary are $\mathcal{O}(Dk^2x\cdot m_s)$.

In the second phase, on each point $Z_1,\ldots,Z_{k-1}$, $\frac{x}{4}$ requests appear in order of distance to the starting point, with requests in between when the it moves over the line.
The latter type of requests induce costs for the adversary of $\mathcal{O}(k\frac{x^2 m_s^2}{m_c})$ if $\frac{m_c}{m_s}$ is sufficiently large (same argument as above).
The costs of the online algorithm can be bounded as in the previous theorem, with the additional argument that while the request moves
past the first potential choice for a $Z_i$, any server covering this segment does not get to the second potential candidate in time.
With this, the costs for the online algorithm are still $\Omega(k x^2 m_s)$.

For the ratio we compare the costs and get $\frac{\Omega(k x^2 m_s)}{\mathcal{O}(Dk^2x\cdot m_s + k\frac{x^2 m_s^2}{m_c})}=\Omega(\frac{m_c}{m_s})$ for sufficiently large $x$.
\end{proof}

\section{Details of Section~\ref{sec:unweighted}}
\label{app:unweighted}

\subsection{Proof of Theorem~\ref{th:simplealgo}}

{
	\renewcommand{\thetheorem}{\ref{th:simplealgo}}

\begin{theorem}
For $k\geq 2$, there are competitive $k$-Server algorithms such that the simple algorithm for the $k$-Mobile Server Problem does not achieve a competitive ratio independent of $n$.
\end{theorem}

}

\begin{proof}
Consider the following instance: All servers and the request start at the same point on the real line. The requests moves $x$ times to the right by a distance of $m_s$ each.
It then moves $y<\frac{x}{4}$ steps to the left again and remains at that point for the remaining $x-2y$ time steps.

An optimal solution may be to just follow the request around with a single server which induces cost $(x+y)m_s$.
Assume the $k$-Server algorithm does the following: As long as the request moves to the right, it gets served by a single server, the requests after that are served by a second server (this $k$-server algorithm would be at most 2-competitive in this instance).
As a result, the online algorithm will move one server to the rightmost point in the sequence and then begin to move a second server towards the request.
When the request has reached its final position, the second server of the online algorithm has moved a distance of $ym_s$ to the right and hence
it takes $x-3y$ more time steps for it come closer than a distance of $ym_s$ to the request.
The server of the online algorithm who followed the request initially to the rightmost point has now a distance of $ym_s$ to the request.
It follows that the costs of the online algorithm are at least $xm_s + (x-3y)ym_s$. By setting $y=\Theta(\sqrt{x})$, the competitive ratio becomes as large as $\Omega(\sqrt{n})$.
\end{proof}

%%%%%%%%%%%%%%%%%%%%%%%%%%%%%%%%%%%%%%%%%%%%%%%%%%%%%%%%%%%%%%%%%%%%%%%%%%%%%%%%%%%%%%%%%%%%%%%%%%%
%%%%%%%%%%%%%%%%%%%%%%%%%%%%%%%%%%%%%%%%%%%%%%%%%%%%%%%%%%%%%%%%%%%%%%%%%%%%%%%%%%%%%%%%%%%%%%%%%%%

\subsection{Proof of Proposition~\ref{prop:projectionunweighted}}

{
	\renewcommand{\thetheorem}{\ref{prop:projectionunweighted}}
	
\begin{proposition}
For the servers $\hat{c}_1,\ldots, \hat{c}_k$ of $\hat{\mathcal{K}}$ it holds $d(\hat{c}_i,r)\leq(8 k + 1)\cdot m_c$ during the whole execution.
The costs of $\hat{\mathcal{K}}$ are at most $\mathcal{O}(k)$ times the costs of $\mathcal{K}$.
\end{proposition}

}

\begin{proof}

We define the following potential: $\phi=\sum_{i=1}^{k}d(c_i,\hat{c}_i)$.
During a phase, the potential decreases every time $\hat{c}_i$ moves to $c_i$ by the same amount $\hat{c}_i$ moves.
Each time $c_i$ moves, $\phi$ increases by at most the amount that $c_i$ moves.

We show that during each phase, $\mathcal{K}$ moves its servers by a total distance of at least $k\cdot m_c$.
Consider the movement of the request from its starting point $r$ to the final point $r'$.
We know that $(4 k + 1)\cdot m_c\geq d(r,r')\geq 4 k\cdot m_c$.
Imagine drawing a straight line between $r$ and $r'$ and separating it into segments of length $m_c$ by hyperplanes orthogonal to the line.
There are now at least $4 k$ such segments.
Since the maximum movement distance of $r$ is $m_c$, there is at least one request per segment.

We consider the configuration of $\mathcal{K}$ at the beginning of the phase.
Every server of \( \mathcal{K} \) has two segments adjacent to its own.
Denote the segments which do not contain a server of \( \mathcal{K} \) and are not adjacent to a segment containing such a server \emph{unoccupied segments}.
Since there are \( 4 k \) segments in total and \( k \) servers of \( \mathcal{K} \), there are at least \( k \) unoccupied segments.
For any unoccupied segment it holds that a server of \( \mathcal{K} \) has to move at least \( m_{c} \) to answer a request in the segment since it needs to cross the entire neighboring segment.
Thus, for at least $k$ segments, the servers of \( \mathcal{K} \) incur costs of at least \( m_{c} \), implying a total movement cost of at least \( k \cdot m_{c} \). 

We can now bound the costs at the end of a phase:
The argument when $c_i\in inner(r)$ is the same as before.
Otherwise, $\phi$ increases by at most $d(\hat{c}_i,\hat{c}_i')\leq (4 k + 1)\cdot m_c$.
This yields $\hat{\mathcal{K}}\leq \mathcal{O}(k)\cdot C_{\mathcal{K}}$.

\end{proof}

\subsection{Proof of Lemma~\ref{le:longtrans}}

{
	\renewcommand{\thetheorem}{\ref{le:longtrans}}

\begin{lemma}
If $\hat{o}\in outer_{t_1}(o^{*})$ at the beginning of a long transition between $t_1$ and $t_2$, then $\hat{o}\in inner_{t_2}(o^{*})$ at the end of the transition.
\end{lemma}

}

\begin{proof}
During the transition time $t^{*}:=t_2-t_1$, $r$ moves a distance of at most $t^{*}\cdot m_c$.
At the beginning, $\hat{o}\in outer_{t_1}(o^{*})$ and $r\in inner_{t_1}(o^{*})$, hence
$d_{t_1}(\hat{o},r) \leq d_{t_1}(\hat{o},o^{*}) + d_{t_1}(o^{*},r)\leq inner_{t_1}(o^{*}) + outer_{t_1}(o^{*})$.
During the first $\lceil inner_{t_1}(o^{*})/m_c\rceil$ time steps, $\hat{o}$ can catch up to $r$ a distance of
$\frac{inner_{t_1}(o^{*})}{m_c}\cdot (1+\frac{1020k}{\delta})\cdot m_c = inner_{t_1}(o^{*}) + \frac{1020k}{\delta}\cdot inner_{t_1}(o^{*})= inner_{t_1}(o^{*}) + outer_{t_1}(o^{*})$
and therefore reaches $r$ (the speed of $\hat{o}$ is an additional $m_c$ higher which accounts for the movement of $r$).
Since $t^{*}>inner_{t_1}(o^{*})/m_c + 2$, there are at least 2 time steps remaining where $\hat{o}$ can move ahead to the final position of $r$.
\end{proof}

\subsection{Proof of Lemma~\ref{le:shorttrans}}

{
	\renewcommand{\thetheorem}{\ref{le:shorttrans}}

\begin{lemma}
Every short transition between $o_i$ in step $t_1$ and $o_j$ in step $t_2$ can increase the distance of some server $s$, which moves at speed at most $(1+\delta)m_s$, to $o^{*}$ by at most
$\min\{6.001\frac{\delta^{2}}{48960k}\cdot d_{t_1}(o^{*},o^{*a}) + 8.001m_c \ ,\ 6.002\cdot \frac{\delta^{2}}{48960k}\cdot d_{t_2}(o^{*},o^{*a}) + 8.002m_c\}$.

Likewise, $s$ decreases its distance to $o^{*}$ by at most $\min\{6.001\frac{\delta^{2}}{48960k}\cdot d_{t_1}(o^{*},o^{*a}) + 8.001m_c \ ,\ 6.002\cdot \frac{\delta^{2}}{48960k}\cdot d_{t_2}(o^{*},o^{*a}) + 8.002m_c\}$.
\end{lemma}

}

\begin{proof}
We consider a short transition from offline server $o_i$ to $o_j$ in between time steps $t_1$ and $t_2$.
By definition, $t^{*}=t_2-t_1\leq \nicefrac{inner_{t_1}(o_i)}{m_c} +2$.

We show that since $o_i$ and $o_j$ must be relatively close together, their distance to the closest server of the online algorithm must be similar.
We first upper bound the distance between $o_i$ and $o_j$ in step $t_1$:
The request travels a distance of at most $t^{*}\cdot m_c$ between the two.
During this time, $o_j$ could have moved a distance of at most $t^{*}\cdot m_s$, and the inner radius could have changed by at most
$t^{*}\cdot\frac{\delta}{16}m_s$.
Since after the $t^{*}$ time steps $r$ enters the inner circle of $o_j$, we can use the above information to trace the distance between the two servers and the inner circle's radius of $o_j$ back to time step $t_1$ (see Figure~\ref{figure:Transition}).

With this knowledge, we get

$\begin{array}{rrcl}
  &d_{t_1}(o_j,o_j^{a}) &\geq& d_{t_1}(o_i,o_i^{a}) - d_{t_1}(o_i,o_j) \\
	  &&\geq& d_{t_1}(o_i,o_i^{a}) - t^{*}\cdot (m_c+m_s+\frac{\delta}{16}m_s) - inner_{t_1}(o_i) - inner_{t_1}(o_j) \\
	  &&\geq& d_{t_1}(o_i,o_i^{a}) - 3\cdot inner_{t_1}(o_i) - inner_{t_1}(o_j) -4m_c\\
		&&\geq& (1-3\cdot \frac{\delta^{2}}{48960k})\cdot d_{t_1}(o_i,o_i^{a}) - \frac{\delta^{2}}{48960k}\cdot d_{t_1}(o_j,o_j^{a}) -4m_c\\
	\Leftrightarrow& d_{t_1}(o_j,o_j^{a}) &\geq& \frac{1-3\cdot \frac{\delta^{2}}{48960k}}{1+ \frac{\delta^{2}}{48960k}}\cdot d_{t_1}(o_i,o_i^{a}) -\frac{4}{1+ \frac{\delta^{2}}{48960k}}\cdot m_c\\
	\Rightarrow& d_{t_1}(o_j,o_j^{a}) &\geq& (1-\frac{4}{48960k+1})\cdot d_{t_1}(o_i,o_i^{a}) - 4m_c. 
\end{array}$

In reverse, we can bound

$\begin{array}{rrcl}
  &d_{t_1}(o_i,o_i^{a}) &\geq& d_{t_1}(o_j,o_j^{a}) - d_{t_1}(o_i,o_j) \\
	  &&\geq& d_{t_1}(o_i,o_i^{a}) - 3\cdot inner_{t_1}(o_i) - inner_{t_1}(o_j) -4m_c\\
		&&\geq& (1-\frac{\delta^{2}}{48960k})\cdot d_{t_1}(o_j,o_j^{a}) - \frac{3\delta^{2}}{48960k}\cdot d_{t_1}(o_i,o_i^{a}) -4m_c\\
	\Leftrightarrow& d_{t_1}(o_i,o_i^{a}) &\geq& \frac{1-\frac{\delta^{2}}{48960k}}{1+ \frac{3\delta^{2}}{48960k}}\cdot d_{t_1}(o_j,o_j^{a}) -\frac{4}{1+ \frac{3\delta^{2}}{48960k}}\cdot m_c\\
	\Rightarrow& d_{t_1}(o_i,o_i^{a}) &\geq& (1-\frac{4}{48960k})\cdot d_{t_1}(o_j,o_j^{a}) - 4m_c. 
\end{array}$

Since $s$ can move away from $o_j$ during the transition and $o_j$ itself moves at speed at most $m_s$, we get

$\begin{array}{rcl}
  d_{t_2}(s,o_j) &\leq& d_{t_1}(s,o_j)  + t^{*}\cdot(2+\delta)m_s \\
	  &\leq& d_{t_1}(s,o_i) + d_{t_1}(o_i,o_j)+ t^{*}\cdot(2+\delta)m_s \\
		&\leq& d_{t_1}(s,o_i) + t^{*}\cdot (m_c+m_s+\frac{\delta}{16}m_s) + inner_{t_1}(o_i) + inner_{t_1}(o_j) + t^{*}\cdot(2+\delta)m_s \\
		&\leq& d_{t_1}(s,o_i) + 5\cdot inner_{t_1}(o_i) + inner_{t_1}(o_j) + 8m_c \\
		&\leq& d_{t_1}(s,o_i) + 5\cdot \frac{\delta^{2}}{48960k}\cdot d_{t_1}(o_i,o_i^{a}) + \frac{\delta^{2}}{48960k}\cdot d_{t_1}(o_j,o_j^{a}) +8m_c.
\end{array}$

To derive the first bound, we get

$\begin{array}{rcl}
  d_{t_2}(s,o_j) &\leq& d_{t_1}(s,o_i) + 5\cdot \frac{\delta^{2}}{48960k}\cdot d_{t_1}(o_i,o_i^{a})
	  + \frac{\delta^{2}}{48960k} \cdot \frac{1}{1-\frac{4}{48960k}}\cdot d_{t_1}(o_i,o_i^{a}) \\
		&&+ (8 +\frac{\delta^{2}}{48960k}\cdot \frac{4}{1-\frac{4}{48960k}} )\cdot m_c \\
		&\leq& d_{t_1}(s,o_i) + 6.001\frac{\delta^{2}}{48960k}\cdot d_{t_1}(o_i,o_i^{a}) + 8.001m_c.
\end{array}$
		
For the second bound, we continue with

$\begin{array}{rcl}
  d_{t_2}(s,o_j) &\leq& d_{t_1}(s,o_i) + 5\cdot \frac{\delta^{2}}{48960k}\cdot d_{t_1}(o_i,o_i^{a}) + \frac{\delta^{2}}{48960k}\cdot d_{t_1}(o_j,o_j^{a}) +8m_c \\
		&\leq& d_{t_1}(s,o_i) + (1+\frac{5}{1-\frac{4}{48960k+1}})\cdot \frac{\delta^{2}}{48960k}\cdot d_{t_1}(o_j,o_j^{a}) +(8+5\cdot \frac{\delta^{2}}{48960k}\cdot\frac{4}{1-\frac{4}{48960k+1}})m_c \\
		&\leq& d_{t_1}(s,o_i) + 6.001\cdot \frac{\delta^{2}}{48960k}\cdot d_{t_1}(o_j,o_j^{a}) + 8.001\cdot m_c.
\end{array}$

\newpage
Next we bound the change in $d(o_j,o_j^{a})$ during the transition:

$\begin{array}{rrcl}
  &d_{t_1}(o_j,o_j^{a}) &\leq& d_{t_2}(o_j,o_j^{a}) + t^{*}\cdot (2+\delta)m_s \\
	  &&\leq& d_{t_2}(o_j,o_j^{a}) + 2\cdot inner_{t_1}(o_i) + 4m_c \\
		&&\leq& d_{t_2}(o_j,o_j^{a}) + 2\cdot \frac{\delta^{2}}{48960k}\cdot d_{t_1}(o_i,o_i^{a}) +4m_c\\
		&&\leq& d_{t_2}(o_j,o_j^{a}) + 2.001\cdot \frac{\delta^{2}}{48960k}\cdot d_{t_1}(o_j,o_j^{a}) + 4.001m_c \\
	\Leftrightarrow& d_{t_1}(o_j,o_j^{a}) &\leq& \frac{1}{1 - 2.001\cdot \frac{\delta^{2}}{48960k}} \cdot d_{t_2}(o_j,o_j^{a}) + 4.002m_c.
\end{array}$

This gives us

$\begin{array}{rcl}
  d_{t_2}(s,o_j) &\leq& d_{t_1}(s,o_i) + \frac{1}{1 - 2.001\cdot \frac{\delta^{2}}{48960k}}\cdot 6.001\cdot \frac{\delta^{2}}{48960k}\cdot d_{t_2}(o_j,o_j^{a}) \\
	&&+ (8.001+ 6.001\cdot \frac{\delta^{2}}{48960k}\cdot 4.002)\cdot m_c \\
  &\leq& d_{t_1}(s,o_i) + 6.002\cdot \frac{\delta^{2}}{48960k}\cdot d_{t_2}(o_j,o_j^{a}) + 8.002m_c.
\end{array}$

For the bound of decreasing the distance, the same proof can be applied:
Start with $d_{t_2}(s,o_j) \geq d_{t_1}(s,o_j)  - t^{*}\cdot(2+\delta)m_s \geq d_{t_1}(s,o_i) - d_{t_1}(o_i,o_j)- t^{*}\cdot(2+\delta)m_s$ and use the same estimations as before from there.
\end{proof}

\subsection{Proof of Lemma~\ref{le:longtermination}}

{
	\renewcommand{\thetheorem}{\ref{le:longtermination}}

\begin{lemma}
Consider a sequence of short transitions which is terminated by a long transition.
If $\hat{o}\in inner(o^{*})$ at the beginning of the sequence, then $\hat{o}\in inner(o^{*})$ after the long transition.
During the sequence of short transitions, $\hat{o}\in outer(o^{*})$.
\end{lemma}

}

\begin{proof}
As in step 2 of the algorithm, we assume the sequence starts at time $t_1$ with $o^{*}=o_i$, and terminates
with a long transition from $o_\ell$ to $o_j$ between time steps $t_2$ and $t_3$.
$\hat{o}$ selects the server $o_\ell$ which passes $r$ on to $o_j$ over the long transition and follows it.
Since $d(o_\ell,o^{*}) \leq outer(o^{*})/3$ for the duration of the sequence, we have $\hat{o}\in outer(o_\ell)$ at the beginning of the sequence
and therefore $\hat{o}\in outer(o_\ell)$ holds for the entire duration.
At the beginning, with
$$\begin{array}{rrcl}
  &d_{t_1}(o^{*},o^{*a}) &\leq& d_{t_1}(o^{*},o_\ell^{a}) \\
	&&\leq& d_{t_1}(o^{*},o_\ell) + d_{t_1}(o_\ell,o_\ell^{a}) \\
	&&\leq& \frac{\delta}{144} \cdot d_{t_1}(o^{*},o^{*a}) + d_{t_1}(o_\ell,o_\ell^{a}) \\
	\Leftrightarrow& d_{t_1}(o^{*},o^{*a}) &\leq& \frac{1}{1-\frac{\delta}{144}}\cdot d_{t_1}(o_\ell,o_\ell^{a})
\end{array}$$
we get
$d_{t_1}(\hat{o},o_\ell) \leq d_{t_1}(\hat{o},o^{*}) + d_{t_1}(o^{*},o_\ell) \leq (\frac{\delta^{2}}{48960k} + \frac{\delta}{144})\cdot d_{t_1}(o^{*},o^{*a})
\leq \frac{1}{1-\frac{\delta}{144}}\cdot(\frac{\delta^{2}}{48960k} + \frac{\delta}{144})\cdot d_{t_1}(o_\ell,o_\ell^{a}) \leq 0.01\cdot\delta\cdot d_{t_1}(o_\ell,o_\ell^{a})$.
Furthermore, since $\hat{o}$ at least holds its relative distance to $o_\ell$, during any step $t$ during the sequence,
$$\begin{array}{rcl}
  d_t(\hat{o},o^{*}) &\leq& d_t(\hat{o},o_\ell) + d_t(o_\ell,o^{*}) \\
	&\leq& \frac{d_{t_1}(\hat{o},o_\ell)}{d_{t_1}(o_\ell,o_\ell^{a})}\cdot d_{t}(o_\ell,o_\ell^{a}) + d_t(o_\ell,o^{*}) \\
	&\leq& 0.01\cdot\delta\cdot d_{t}(o_\ell,o_\ell^{a}) + \frac{\delta}{144}\cdot d_t(o^{*},o^{*a}) \\
	&\leq& 0.01\cdot\delta\cdot d_{t}(o_\ell,o^{*a}) + \frac{\delta}{144}\cdot d_t(o^{*},o^{*a}) \\
	&\leq& 0.01\cdot\delta\cdot (d_{t}(o_\ell,o^{*}) + d_{t}(o^{*},o^{*a})) + \frac{\delta}{144}\cdot d_t(o^{*},o^{*a}) \\
	&\leq& 0.01\cdot\delta\cdot (\frac{\delta}{144}\cdot d_{t}(o^{*},o^{*a}) + d_{t}(o^{*},o^{*a})) + \frac{\delta}{144}\cdot d_t(o^{*},o^{*a}) \\
	&\leq& \frac{\delta}{48}\cdot d_t(o^{*},o^{*a})
\end{array}$$
and therefore $\hat{o}\in outer_t(o^{*})$ during the whole sequence.
By Lemma~\ref{le:longtrans}, we have $\hat{o}\in inner(o^{*})$ after the long transition.
\end{proof}

\newpage

\subsection{Proof of Lemma~\ref{le:shorttermination}}

{
	\renewcommand{\thetheorem}{\ref{le:shorttermination}}

\begin{lemma}
Consider a sequence of short transitions which is terminated by a short transition from $o_\ell$ to $o_j$, where at one point prior in the sequence $d(o_j,o^{*}) > outer(o^{*})/3$.
If $\hat{o}\in inner(o^{*})$ at the beginning of the sequence and $d(o^{*},o^{*a}) \geq 51483\frac{km_c}{\delta^2}$ at all times, then $\hat{o}\in inner(o^{*})$ after the transition to $o_j$.
During the sequence, $\hat{o}\in outer(o^{*})$.
\end{lemma}

}

\begin{proof}
We assume the sequence starts at time $t_1$ with $o^{*}=o_i$, and terminates
with a short transition from $o_\ell$ to $o_j$ between time steps $t_2$ and $t_3$.

We first consider the case that $\hat{o}$ would run outside $outer(o^{*})$ if it moved directly to its target point.
First, we need to show that $d(\hat{o},o_\ell)=\frac{2\delta}{145}\cdot d(o_\ell,o_\ell^{a})$ is reached before this happens.
In the beginning, it holds $d(\hat{o},o_\ell)\leq\frac{2\delta}{145}\cdot d(o_\ell,o_\ell^{a})$:
$d_{t_1}(\hat{o},o_\ell) \leq d_{t_1}(\hat{o},o^{*}) + d_{t_{1}}(o^{*},o_\ell) \leq (\frac{\delta^{2}}{48960k} + \frac{\delta}{144})\cdot d_{t_1}(o^{*},o^{*a})$.
With
$$\begin{array}{rrcl}
  & d(o^{*},o^{*a}) &\leq& d(o^{*},o_\ell) + d(o_\ell,o_\ell^{a}) \\
	&&\leq& \frac{\delta}{144}\cdot d(o^{*},o^{*a})+ d(o_\ell,o_\ell^{a}) \\
	\Leftrightarrow& (1-\frac{\delta}{144})\cdot d(o^{*},o^{*a}) &\leq& d(o_\ell,o_\ell^{a})
\end{array}$$
we get $d_{t_1}(\hat{o},o_\ell) \leq\frac{1}{1-\frac{\delta}{144}}\cdot (\frac{\delta^{2}}{48960k} + \frac{\delta}{144})\cdot d_{t_1}(o_\ell,o_\ell^{a})
<\frac{2\delta}{145}\cdot d_{t_1}(o_\ell,o_\ell^{a})$.

Now assume $d(\hat{o},o_\ell) \leq \frac{2\delta}{145}\cdot d(o_\ell,o_\ell^{a})$.
Then
$$\begin{array}{rcl}
  d(\hat{o},o^{*}) &\leq& d(\hat{o},o_\ell) + d(o_\ell,o^{*}) \\
	&\leq& \frac{2\delta}{145}\cdot d(o_\ell,o_\ell^{a}) + \frac{\delta}{144}\cdot d(o^{*},o^{*a}) \\
	&\leq& \frac{2\delta}{145}\cdot (d(o_\ell,o^{*}) + d(o^{*},o^{*a})) + \frac{\delta}{144}\cdot d(o^{*},o^{*a}) \\
	&\leq& \frac{2\delta}{145}\cdot (1+\frac{\delta}{144})\cdot d(o^{*},o^{*a})+ \frac{\delta}{144}\cdot d(o^{*},o^{*a}) \\
	&\leq& \frac{\delta}{48}\cdot d(o^{*},o^{*a}),
\end{array}$$
meaning $\hat{o}\in outer(o^{*})$ for the duration of the sequence.
Taking the negation of that statement it also follows that $d(\hat{o},o_\ell)=\frac{2\delta}{145}\cdot d(o_\ell,o_\ell^{a})$ is reached before $\hat{o}\notin outer(o^{*})$.

Note that $\hat{o}$ can maintain the point at the fixed distance to $o_\ell$ which is closest to the final position of $o_j$:
Imagine the radius $\frac{2\delta}{145}\cdot d(o_\ell,o_\ell^{a})$ stays fixed and only $o_\ell$ moves by at most $m_s$.
Then the point at the fixed radius closest to $o_j^{(t_3)}$ only changes by at most $m_s$.
Afterwards the radius changes by at most $3m_s\cdot \frac{2\delta}{145}<\frac{\delta}{20}m_s$ and hence the movement speed of $(1+\frac{\delta}{8})m_s$ is sufficient.

We now need to determine that at the final time step $t_3$, $d_{t_3}(o_j,o_\ell) \leq \frac{2\delta}{145}\cdot d_{t_3}(o_\ell,o_\ell^{a})$.
Apply Lemma~\ref{le:shorttrans} by setting $s=o_\ell$ and we can bound $d_{t_3}(o_j,o_\ell) \leq 6.002\cdot \frac{\delta^{2}}{48960k}\cdot d_{t_3}(o^{*},o^{*a}) + 8.002m_c
\leq \frac{1}{1-\frac{\delta}{144}}\cdot (\frac{6.002\delta^{2}}{48960k} + \frac{8.002\delta^2}{43170}) \cdot d_{t_3}(o_\ell,o_\ell^{a}) < \frac{2\delta}{145}\cdot d_{t_3}(o_\ell,o_\ell^{a})$.

\medskip

Now assume it holds true that $\hat{o}$ can move straight towards the final position of $o_j$ with speed $(1+\frac{\delta}{8})m_s$ without ever leaving $outer(o^{*})$.
In this case, we compute a path which constitutes an upper bound on the distance $\hat{o}$ has to traverse, using the following definition:

\begin{figure}[ht]
	\centering
	\includegraphics[page=4, width= 0.8\textwidth, clip=true, trim = 2cm 8.5cm 11.5cm 3.25cm]{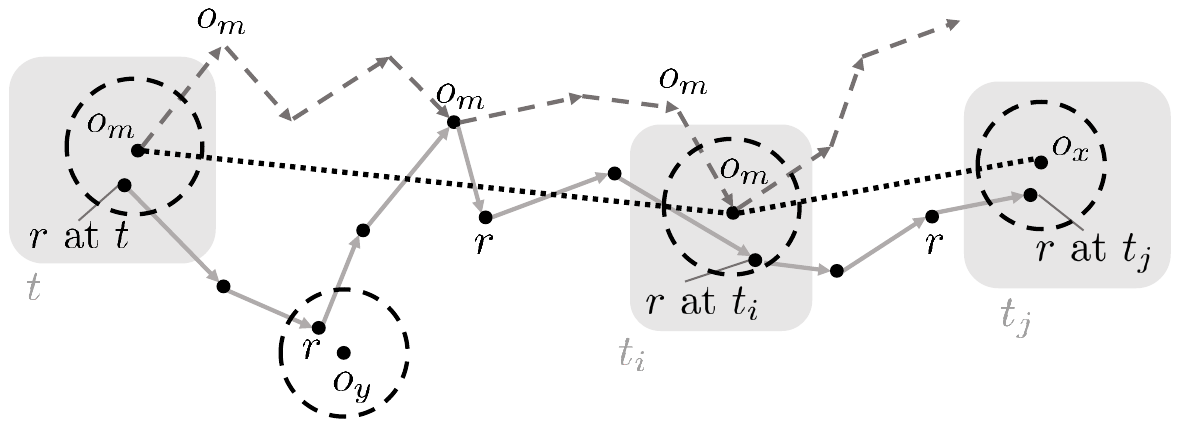}
	\caption{The construction of a transition path.
	The transition path is marked by black points, while the movement of \( o_{m} \) is depicted by dashed arrows.
	The movement of \( r \) is marked by the gray arrows.
	Starting at the position of \( o_{m} \) at \( t \), the last time step \( t_{i} \) is identified at which \( o_{m}=o^{*} \) and $r\in inner(o_m)$. 
	Note, that the role of \( o^{*} \) might change multiple times between \( t \) and \( t_{i} \).
	}
	\label{figure:transition-path}
\end{figure}

\begin{definition}[Transition Path]\label{definition:transition-path}
	Assume \( o_{m}=o^{*} \) at time step \( t \) and \(o_{n} =o^{*} \) at some later time step \( t' \).
	Consider the path constructed as follows.
	Start at the position of \( o_{m} \) in time step \( t \).
	Let \( t_{i} \) be the last time step before \( t' \) in which \( o_{m}=o^{*} \) and $r\in inner(o_m)$.
	The first part of the path goes from \( o_{m} \)'s position at time step \( t \) to \( o_{m} \)'s position in time step \( t_{i} \).
	Afterwards, a short transition from $o_{m}$ to some other server $o_{x}$ between time step $t_{i}$ and $t_{j}$ occurs, in which case our path goes from $o_{m}$ in $t_{i}$ to $o_{x}$ in $t_{j}$.
	Continue the procedure recursively until \( o_{n} \) in time step $t'$ is reached.
	We call the constructed path a \( (t,t') \)-\emph{transition path}.
	See \cref{figure:transition-path} for an illustration of one of the recursion steps.
\end{definition}

Now consider the \( (t_{1},t_{3}) \)-transition path.
The distance traveled by \( \hat{o} \) is bounded by the distance of \( \hat{o} \) to \( o^{*} \) at time \( t_{1} \) plus the length of the transition path.
The former has a length of $d_{t_1}(\hat{o},o^{*})\leq \frac{\delta^{2}}{48960k}\cdot d_{t_1}(o^{*},o^{*a})$.

To upper bound the length of the \( (t_{1},t_{3}) \)- transition path, we divide it into two types of edges (excluding the first edge):
The first type is between the same offline server in different time steps.
If the total time is $\hat{t}=t_3-t_1$, the maximum distance induced is $\hat{t}\cdot m_s$.

The second type of edges are between different offline servers and represent a short transition.
By construction, there are at most $k$ such edges.
With the help of Lemma~\ref{le:shorttrans} we may upper bound the length of an edge by $6.001\cdot \frac{\delta^{2}}{48960k}\cdot d_{t'}(o^{*},o^{*a}) + 8.001m_c$,
where $t'$ is the time the transition begins (in the lemma, set $s$ to a static server at the position of the server who passes the request at time $t'$).

The distance $d(o^{*},o^{*a})$ can change in two ways over time:
It changes due to the movement of the servers or due to a role change of $o^{*}$, where it suffices to consider only those short transitions included in our constructed path.
Let $t'_1,\ldots, t'_k$ be the points in time where the short transitions inducing the second type edges begin.
We can upper bound their total length as
$
\sum_{i=1}^{k} (6.001\cdot \frac{\delta^{2}}{48960k}\cdot d_{t'_i}(o^{*},o^{*a}) + 8.001m_c)
$.
Assuming the highest possible distance for each of the $d_{t'_i}(o^{*},o^{*a})$, we get
for the first transition the total distance of the movement during the sequence added to the original length, which is $d_{t_1}(o^{*},o^{*a}) + \hat{t}\cdot (2+\delta)m_s$.
The transitions after that build inductively on the resulting lengths.
Define $T_0:=d_{t_1}(o^{*},o^{*a}) + \hat{t}\cdot(2+\delta) m_s$.
The first edge length is upper bounded by $A_1:=\frac{6.001\delta^{2}}{48960k}\cdot T_0 + 8.001 m_c$, the resulting value for $d(o^{*},o^{*a})$ is $T_1:=T_0+A_1$.
In general, $A_i:=\frac{6.001\delta^{2}}{48960k}\cdot T_{i-1} + 8.001 m_c$ and $T_i:=T_{i-1}+A_i=T_0+\sum_{j=1}^{i}A_j$.
We can bound the total increase by
$$\begin{array}{rrcl}
  &\sum_{i=1}^{k}A_i &=& \sum_{i=1}^{k}\left(\frac{6.001\delta^{2}}{48960k}\cdot (T_0+\sum_{j=1}^{i-1}A_j) + 8.001 m_c\right) \\
	&&\leq& k\cdot\frac{6.001\delta^{2}}{48960k}\cdot T_0 + k\cdot 8.001m_c + k\cdot \frac{6.001\delta^{2}}{48960k}\cdot \sum_{j=1}^{k}A_j \\
	\Leftrightarrow& (1-\frac{6.001\delta^{2}}{48960})\cdot \sum_{i=1}^{k}A_i &\leq& \frac{6.001\delta^{2}}{48960}\cdot T_0 + 8.001 k m_c \\
	\Rightarrow& \sum_{i=1}^{k}A_i &\leq& 0.0002\delta^2\cdot T_0 + 8.002 k m_c.
\end{array}$$

The total path length may hence bounded by $\hat{t}\cdot m_s + 0.0002\delta^2\cdot (d_{t_1}(o^{*},o^{*a}) + \hat{t}\cdot (2+\delta)m_s) + 8.002km_c + \frac{\delta^{2}}{48960k}\cdot d_{t_1}(o^{*},o^{*a})$.

\medskip

For comparison, we lower bound the time it takes $o_j$ to move into position such that a short transition can occur.
Take a time step $t$ where $o_j\notin outer_{t}(o^{*})/3 \Rightarrow d_{t}(o_j,o^{*})>\frac{\delta}{144}\cdot d_{t}(o^{*},o^{*a})$.
We may assume that $t=t_1$, otherwise the travel time for $o_j$ simply increases.
For a short transition between time steps $t_2$ and $t_3$ to $o_j$ to occur, we need $r\in inner_{t_2}(o_\ell)$, $r\in inner_{t_3}(o_j)$ and $t^{*}:=t_3-t_2\leq inner_{t_2}(o_\ell)/m_c + 2$.
We have $d_{t_2}(o_j,o_\ell)\leq t^{*}\cdot (m_c+m_s+\frac{\delta}{16}m_s) + inner_{t_2}(o_\ell) + inner_{t_2}(o_j)$ (see Figure~\ref{figure:Transition} and the proof of Lemma~\ref{le:shorttrans}).

With $d_{t_2}(o_j,o_j^{a}) \leq d_{t_2}(o_j,o_\ell) + d_{t_2}(o_\ell,o_\ell^{a})$
we get
$$\begin{array}{rrcl}
  & d_{t_2}(o_j,o_\ell) &\leq& inner_{t_2}(o_j) + inner_{t_2}(o_\ell) + t^{*}\cdot 2m_c  \\
	&&\leq& \frac{\delta^{2}}{48960k}\cdot d_{t_2}(o_j,o_j^{a}) + 3\cdot inner_{t_2}(o_\ell) + 4m_c \\
	&&\leq& 4\cdot \frac{\delta^{2}}{48960k}\cdot d_{t_2}(o_\ell,o_\ell^{a}) + \frac{\delta^{2}}{48960k}\cdot d_{t_2}(o_j,o_\ell) +4m_c \\
	\Leftrightarrow& (1-\frac{\delta^{2}}{48960k})\cdot d_{t_2}(o_j,o_\ell) &\leq& 4\cdot \frac{\delta^{2}}{48960k}\cdot d_{t_2}(o_\ell,o_\ell^{a}) + 4m_c \\
	\Rightarrow& d_{t_2}(o_j,o_\ell) &\leq& 4.001\cdot \frac{\delta^{2}}{48960k}\cdot d_{t_2}(o_\ell,o_\ell^{a}) + 4.001m_c.
\end{array}$$

Comparing the distances at $t_1$ and $t_2$, we conclude that
$d_{t_1}(o_j,o^{*}) - d_{t_2}(o_j,o^{*})\geq \frac{\delta}{144}\cdot d_{t_1}(o^{*},o^{*a}) - 4.001\cdot \frac{\delta^{2}}{48960k}\cdot d_{t_2}(o^{*},o^{*a}) - 4.001m_c$.

In order to lower bound the number of time steps $\hat{t}:=t_2-t_1$ needed for bridging that distance, we first examine the change in $d(o^{*},o^{*a})$.
Recall that $o^{*}=o_i$ in $t_1$ and $o^{*}=o_\ell$ in $t_2$.
We can represent the movement of $o^{*}$ with the $(t_1,t_2)$-transition path.
The distance $d(o^{*},o^{*a})$ can change in two ways over time:
It changes due to the movement of the servers or due to a role change of $o^{*}$, where it suffices to consider only those short transitions included in our constructed path.
If we set the beginnings of the short transitions at time steps $t'_1,\ldots, t'_{k}$, we get the upper bound similar to before:

$$\begin{array}{rcl}
d_{t_2}(o^{*},o^{*a}) &\leq& d_{t_1}(o^{*},o^{*a}) + \hat{t}\cdot (2+\delta)m_s + \sum_{i=1}^{k} (6.001\cdot \frac{\delta^{2}}{48960k}\cdot d_{t'_i}(o^{*},o^{*a}) + 8.001m_c) \\
&\leq& d_{t_1}(o^{*},o^{*a}) + \hat{t}\cdot (2+\delta)m_s + 0.0002\delta^2\cdot (d_{t_1}(o^{*},o^{*a}) + \hat{t}\cdot (2+\delta)m_s) + 8.002km_c
\end{array}$$
Continuing from above, we have
$$\begin{array}{rcl}
  d_{t_1}(o_j,o^{*}) - d_{t_2}(o_j,o^{*}) &\geq& \frac{\delta}{144}\cdot d_{t_1}(o^{*},o^{*a}) - 4.001\cdot \frac{\delta^{2}}{48960k}\cdot d_{t_2}(o^{*},o^{*a}) - 4.001m_c \\
	&\geq& \frac{\delta}{144}\cdot d_{t_1}(o^{*},o^{*a}) - \frac{4.001\delta^{2}}{48960k}\cdot ( 1.0002\cdot (d_{t_1}(o^{*},o^{*a}) + \hat{t}\cdot (2+\delta)m_s)\\
	&&  + 8.002km_c) - 4.001m_c.
\end{array}$$

Now we consider the ways in which $d(o_j,o^{*})$ shrinks:
The first is the movement of $o_j$ and $o^{*}$, reducing the distance by at most $2m_s$ per time step, i.e., if the entire sequence lasts $\hat{t}$ steps, the maximum reduction is $\hat{t}\cdot 2m_s$.
The other way is by the role change of $o^{*}$.
Note that above, we just accounted for the change of the distance $d(o^{*},o^{*a})$ due to the role change, and not for the change of $d(o_j,o^{*})$.
Lemma~\ref{le:shorttrans} gives us that the distance of $o_j$ to any server decreases by at most $6.001\cdot \frac{\delta^{2}}{48960k}\cdot d_{t}(o^{*},o^{*a}) + 8.001m_c$.
This decrease is maximized the same as above, i.e., $0.0002\delta^2\cdot (d_{t_1}(o^{*},o^{*a}) + \hat{t}\cdot 2m_s) + 8.002 km_c$.

We can now lower bound the number of time steps it takes to complete the sequence:
It is bounded by the minimum time $\hat{t}$, such that

$\begin{array}{ccl}
  && \hat{t}\cdot 2m_s + 0.0002\delta^2\cdot (d_{t_1}(o^{*},o^{*a}) + \hat{t}\cdot 2m_s) + 8.002km_c \\
	  &\geq& \frac{\delta}{144}\cdot d_{t_1}(o^{*},o^{*a}) - \frac{4.001\delta^{2}}{48960k}\cdot (d_{t_1}(o^{*},o^{*a}) + \hat{t}\cdot (2+\delta)m_s \\
		&& + 1.0002\cdot (d_{t_1}(o^{*},o^{*a}) + \hat{t}\cdot (2+\delta)m_s) + 8.002km_c) - 4.001m_c \\
	\Leftrightarrow& &\hat{t}\cdot 2.0004m_s + \frac{4.001\delta^{2}}{48960k}\cdot 2.0002\cdot \hat{t}\cdot (2+\delta)m_s \\
	&\geq& \frac{\delta}{144}\cdot d_{t_1}(o^{*},o^{*a}) - \frac{4.001\delta^{2}}{48960k}\cdot 2.0002\cdot d_{t_1}(o^{*},o^{*a}) - 0.0002\delta^2\cdot d_{t_1}(o^{*},o^{*a}) \\
	  && - 8.002km_c - \frac{4.001\delta^{2}}{48960k}\cdot 8.002km_c - 4.001m_c \\
	\Rightarrow&& 2.0009 \cdot \hat{t}\cdot m_s 
	\geq 0.0065\delta \cdot d_{t_1}(o^{*},o^{*a}) -12.0047km_c.
\end{array}$

\newpage

To finish the proof, we show that $\hat{o}$ has enough time to reach its destination by comparing the lower bound of the time $o_j$ takes to move into position
to the upper bound of the travel path of $\hat{o}$:
$$\begin{array}{rrcl}
  & \hat{t}\cdot(1+\frac{\delta}{8})\cdot m_s &\geq& \hat{t}\cdot m_s +  0.0002\delta^2\cdot (d_{t_1}(o^{*},o^{*a}) + \hat{t}(2+\delta)\cdot m_s)\\
	  &&&+ 8.002km_c + \frac{\delta^{2}}{48960k}\cdot d_{t_1}(o^{*},o^{*a}) \\
	\Leftrightarrow& \hat{t}\cdot(1+\frac{\delta}{8})\cdot m_s 
	- (1+0.0006\delta^2)\cdot \hat{t}\cdot m_s &\geq& (0.0002\delta^2+ \frac{\delta^{2}}{48960k})\cdot d_{t_1}(o^{*},o^{*a})+ 8.002km_c \\
	\Leftarrow& \hat{t}\cdot (\frac{\delta}{8} - 0.0006\delta^2)m_s &\geq& (0.0002\delta^2+ \frac{\delta^{2}}{48960k})\cdot d_{t_1}(o^{*},o^{*a})+ 8.002km_c \\
	\Leftarrow& \frac{1}{2.0009 \cdot m_s}\cdot (0.0065\delta \cdot d_{t_1}(o^{*},o^{*a})\\ &-12.0047km_c)
	\cdot (\frac{\delta}{8} - 0.0006\delta^2)m_s
	&\geq& (0.0002\delta^2+ \frac{\delta^{2}}{48960k})\cdot d_{t_1}(o^{*},o^{*a})+ 8.002km_c \\
	\Leftarrow& 0.0004\delta^2\cdot d_{t_1}(o^{*},o^{*a}) - 0.75\delta km_c &\geq& (0.0002\delta^2+ \frac{\delta^{2}}{48960k})\cdot d_{t_1}(o^{*},o^{*a})+ 8.002km_c \\
	\Leftarrow& 0.00017\delta^2\cdot d_{t_1}(o^{*},o^{*a}) &\geq& (8.002 + 0.75\delta)km_c \\
	\Leftarrow& d_{t_1}(o^{*},o^{*a}) &\geq& 51483k\frac{m_c}{\delta^2}
\end{array}$$
\end{proof}

\subsection{Proof of Lemma~\ref{le:helperdistance}}

{
	\renewcommand{\thetheorem}{\ref{le:helperdistance}}

\begin{lemma}
During the execution of the algorithm, $d(\hat{a},\hat{o}) \leq 2\cdot d(o^{*},o^{*a}) + d(a^{*},r)$ as long as the algorithm is in step 1 or 2.
\end{lemma}

}

\begin{proof}
We argue that $\hat{o}\in outer(o^{*})$ or $\hat{o}=r$.
We have demonstrated, that during a sequence of short transitions, $\hat{o}$ never leaves $outer(o^{*})$.
It remains to show that the statement holds during a long transition.
We observe $\hat{o}$ during the transition time $t^{*}=t_2-t_1$.
Before the first step, $r\in inner_{t_1}(o^{*})$ and $\hat{o}\in outer_{t_1}(o^{*})$.
We have already shown that for $t^{*} \geq inner_{t_1}(o^{*})/m_c$, $\hat{o}$ catches up to $r$ within the time $t^{*}$ in the proof of Lemma~\ref{le:longtrans}.
Expressed in distance, $\hat{o}$ catches up to $r$ when $r$ is a distance of $inner_{t_1}(o^{*})$ outside the inner circle of $o^{*}$.
We show that at this time, $r$ is still in $outer(o^{*})$:
Let $\hat{t}=\left\lceil inner_{t_1}(o^{*})/m_c\right\rceil$.
We have

$\begin{array}{rcl}
  d_{t_1+\hat{t}}(r,o^{*}) &\leq& d_{t_1}(r,o{*}) + \hat{t}\cdot 2m_c \\
	&\leq& \frac{\delta^2}{48960k}\cdot d_{t_1}(o^{*},o^{*a}) + 2\cdot inner_{t_1}(o^{*}) + 2m_c \\
	&\leq& 3\cdot \frac{\delta^2}{48960k}\cdot d_{t_1}(o^{*},o^{*a}) + 2m_c.
\end{array}$

With

$\begin{array}{rrcl}
  &d_{t_1+\hat{t}}(o^{*},o^{*a}) &\geq& d_{t_1}(o^{*},o^{*a}) - \hat{t}\cdot (2+\delta)m_s \\
	&&\geq& d_{t_1}(o^{*},o^{*a}) - 2\cdot inner_{t_1}(o^{*}) - 2m_c \\
	\Leftrightarrow& d_{t_1+\hat{t}}(o^{*},o^{*a}) + 2m_c &\geq& (1-\frac{2\delta^2}{48960k})\cdot d_{t_1}(o^{*},o^{*a}) \\
	\Rightarrow& 2\cdot d_{t_1+\hat{t}}(o^{*},o^{*a}) + 4m_c &\geq& d_{t_1}(o^{*},o^{*a})
\end{array}$

we get $d_{t_1+\hat{t}}(r,o^{*}) \leq \frac{6\delta^2}{48960k}\cdot d_{t_1+\hat{t}}(o^{*},o^{*a}) + 3m_c \leq \frac{\delta}{48}\cdot d_{t_1+\hat{t}}(o^{*},o^{*a})$ as long as $d_{t_1+\hat{t}}(o^{*},o^{*a})>145m_c$.
This implies that at all times, either $\hat{o}\in outer(o^{*})$ or $r=\hat{o}$.

We now turn to the claim of the lemma.
If $\hat{o}\in outer(o^{*})$, then $d(\hat{a},\hat{o}) \leq d(o^{*a},\hat{o}) \leq 2\cdot d(o^{*},o^{*a})$.
If $\hat{o}=r$, then $\hat{a}=a^{*}$ and therefore $d(\hat{a},\hat{o})=d(a^{*},r)$.
\end{proof}

\subsection{Proof of Lemma~\ref{le:stepthree}}

{
	\renewcommand{\thetheorem}{\ref{le:stepthree}}

\begin{lemma}
After the execution of step 3 it holds $\hat{o}=r$.
Furthermore, $d(\hat{a},\hat{o}) \leq 2\cdot d(o^{*},o^{*a}) + d(a^{*},r)$ during step 3 of the algorithm.
\end{lemma}

}

\begin{proof}
We define time steps $t_1$ and $t_2$ such that they encompass step 3 of the algorithm, i.e., $t_1$ and $t_2$ are chosen minimal such that
$d_{t_1}(o^{*},o^{*a}) < 51483\frac{m_c}{\delta^2}$ and $d_{t_2}(o^{*},o^{*a}) \geq 2\cdot 51483\frac{m_c}{\delta^2}$.
Since $d(o^{*},o^{*a})$ changes by at most $(2+\delta)m_s\leq 2m_c$ in each time step, $t_2-t_1\geq 25741.5\cdot\frac{1}{\delta^2}$.

If at time $t_1$, the procedure is in a long transition, the algorithm already follows $r$ and can continue as usual (the result for the long transition holds independently of $d(o^{*},o^{*a})$).
Otherwise, we have $\hat{o},r\in outer(o^{*})$.
Hence $d_{t_1}(\hat{o},r)\leq \frac{\delta}{24}\cdot d_{t_1}(o^{*},o^{*a}) \leq \frac{51483}{24}\cdot\frac{m_c}{\delta}$.
The server $\hat{o}$ catches up to $r$ a distance of at least $(1+\frac{1020k}{\delta})\cdot m_c$ per time step.
Clearly, $(t_2-t_1)\cdot (1+\frac{1020k}{\delta})\cdot m_c > \frac{51483}{24}\cdot\frac{m_c}{\delta}$ and therefore $\hat{o}=r$ at time $t_2$.

The second claim, $d(\hat{a},\hat{o}) \leq 2\cdot d(o^{*},o^{*a}) + d(a^{*},r)$ can be shown the same way as in the previous lemma, where it is clear that $r$ is reached before
$d(o^{*},o^{*a})$ falls below $145m_c$.
\end{proof}

\subsection{Proof of Lemma~\ref{le:easycancel}}

{
	\renewcommand{\thetheorem}{\ref{le:easycancel}}

\begin{lemma}
$\Delta\psi\leq Y\cdot\frac{m_c}{\delta m_s}\cdot C_{\mathcal{K}} - \sum_{i=1}^{k}d(a_i,a_i')$.
\end{lemma}

}

\begin{proof}
  Assume that $a^{*}=a_1$. Every other server $a_i$ moves towards its counterpart $c_i$, hence
	
  $\begin{array}{rcl}
	  \Delta\psi &\leq& Y\cdot\frac{m_c}{\delta m_s}\sum\limits_{i=1}^{k}(d(a_i',c_i')-d(a_i,c_i))\\
	  &\leq& Y\cdot\frac{m_c}{\delta m_s}\left(d(a_1',c_1')-d(a_1,c_1) + \sum\limits_{i=2}^{k}(d(c_i,c_i')-d(a_i,a_i'))\right).\\
	\end{array}$
	
	Now, if $\mathcal{K}$ serves the request with $c_1$, i.e., $c_1'=r'$, then
	
	$$\Delta\psi \leq Y\cdot\frac{m_c}{\delta m_s} \sum\limits_{i=i}^{k}(d(c_i,c_i')-d(a_i,a_i')).$$
	
	Otherwise, $\mathcal{K}$ serves the request with another server (assume $c_2$).
	Since $a_2$ was not chosen as $a^{*}$, it moves the full distance of $(1+\delta)m_s$ and hence
	
	$\begin{array}{rcl}
	  \Delta\psi &\leq& Y\cdot\frac{m_c}{\delta m_s}\left(d(a_1,a_1') + d(c_1,c_1') + d(c_2,c_2')-d(a_2,a_2') + \sum\limits_{i=3}^{k}(d(c_i,c_i')-d(a_i,a_i'))\right) \\
		&\leq& Y\cdot\frac{m_c}{\delta m_s}\left(\sum\limits_{i=1}^{k}d(c_i,c_i') - \frac{\delta}{2}m_s - \sum\limits_{i=3}^{k}d(a_i,a_i')\right).
	\end{array}$
	
	The lemma follows by setting $Y\geq 8$, as $d(a_1,a_1') + d(a_2,a_2')\leq 4 m_s$.
\end{proof}

\subsection{Proof of Lemma~\ref{le:massivecancel}}

{
	\renewcommand{\thetheorem}{\ref{le:massivecancel}}

\begin{lemma}
 If $d(a^{*'},r')>0$, then $\Delta\psi\leq Y\cdot\frac{m_c}{\delta m_s}C_\mathcal{K} - \sum\limits_{i=1}^{k}d(a_i,a_i') - \frac{Y-4}{2}m_c.$
\end{lemma}

}

\begin{proof}
We assume $a^{*}=a_1$.
Since $d(a^{*'},r')>0$, we have $d(a_1,a_1')=(1+\frac{\delta}{2})m_s$.
If $r$ is served by $c_1$, then
$$\begin{array}{rcl}
  \Delta\psi &=& Y\cdot\frac{m_c}{\delta m_s}\sum\limits_{i=1}^{k}(d(a_i',c_i') - d(a_i,c_i)) \\
	&\leq& Y\cdot\frac{m_c}{\delta m_s}\sum\limits_{i=1}^{k}(d(c_i,c_i') - d(a_i,a_i')) \\
	&\leq& Y\cdot\frac{m_c}{\delta m_s}C_\mathcal{K} - \frac{m_c}{\delta m_s}\sum\limits_{i=1}^{k}d(a_i,a_i') - (Y-1)\cdot\frac{m_c}{\delta m_s}(1+\frac{\delta}{2})m_s.
\end{array}$$
If $r$ is served by a different server of $\mathcal{K}$ (assume $c_2$), then 
$$\begin{array}{rcl}
  \Delta\psi &=& Y\cdot\frac{m_c}{\delta m_s}\sum\limits_{i=1}^{k}(d(a_i',c_i') - d(a_i,c_i)) \\
	&\leq& Y\cdot\frac{m_c}{\delta m_s}\sum\limits_{i=1}^{k}d(c_i,c_i') - \frac{m_c}{\delta m_s}\sum\limits_{i=3}^{k}d(a_i,a_i') - Y\cdot\frac{m_c}{\delta m_s}\cdot\frac{\delta m_s}{2} \\
	&\leq& Y\cdot\frac{m_c}{\delta m_s}C_\mathcal{K} - \sum\limits_{i=1}^{k}d(a_i,a_i') - \frac{Y-4}{2}m_c.
\end{array}$$

This term is larger then the former one for sufficiently large $Y$.
\end{proof}

\subsection{Proof of Lemma~\ref{le:phicancel}}

{
	\renewcommand{\thetheorem}{\ref{le:phicancel}}

\begin{lemma}
If $d(a^{*},r')> 102970\frac{k m_c}{\delta^2}$ and $r'\in inner(o^{*'})$, then $d(a_i',\hat{o}')-d(a_i,\hat{o})\leq -\frac{\delta}{8}m_s$.
\end{lemma}

}

\begin{proof}
By our construction of the simulated $k$-Server algorithm, we have $d(c_i',r') \leq 9km_c \leq \frac{\delta^2}{9724}\cdot d(a^{*'},r')$.
Furthermore,
$$\begin{array}{rrcl}
  & d(o^{*'},o^{*a'}) &\leq& d(o^{*'},a^{*'}) \\
	&&\leq& d(o^{*'},r') + d(r',a^{*'}) \\
	\Leftrightarrow& (1-\frac{\delta^2}{48960k})\cdot d(o^{*'},o^{*a'}) &\leq& d(r',a^{*'}).
\end{array}$$
Hence
$$\begin{array}{rcl}
  d(c_i',\hat{o}') &\leq& d(c_i',r') + d(r',o^{*'}) + d(o^{*'},\hat{o}') \\
	&\leq& \frac{\delta^2}{9724}\cdot d(a^{*'},r') + (\frac{\delta}{48} + \frac{\delta^2}{48960k})\cdot d(o^{*'},o^{*a'}) \\
	&\leq& 0.021\delta\cdot d(a^{*'},r') \\
	&\leq& 0.021\delta\cdot d(a_i',r')
\end{array}$$
and with Lemma~\ref{le:Geo}, we get $d(a_i',\hat{o}')-d(a_i,\hat{o}')\leq -\frac{1+\frac{1}{2}\delta}{1+\delta}d(\hat{a},\hat{a}')$.

In order to bound the movement of $\hat{o}$, we need to show that $d(o^{*},o^{*a}) \geq 2\cdot 51483\frac{km_c}{\delta^2}$.
We use

$\begin{array}{rrcl}
  &d(a^{*},r') &\leq& m_c + d(a^{*'},r') \\
	&&\leq& m_c + d(o^{*a'},r') \\
	&&\leq& m_c + (1 + \frac{\delta^{2}}{48960k}) \cdot d(o^{*'},o^{*a'}) \\
	\Leftrightarrow& \frac{1}{1 + \frac{\delta^{2}}{48960k}}(d(a^{*},r')-m_c) &\leq& d(o^{*'},o^{*a'})
\end{array}$

The bound follows from $d(a^{*},r')> 102970\frac{k m_c}{\delta^2}$.

From Proposition~\ref{th:invariant} we get $d(\hat{o},\hat{o}')\leq (1+\frac{\delta}{8})m_s$ and therefore $d(\hat{a}',\hat{o}') - d(\hat{a},\hat{o})\leq -\frac{1+\frac{1}{2}\delta}{1+\delta}d(\hat{a},\hat{a}') + d(\hat{o},\hat{o'}) \leq -(1+\frac{\delta}{8})m_s  + (1+\frac{\delta}{8})m_s \leq -\frac{\delta}{8}m_s$.

\end{proof}

\subsection{Proof of Lemma~\ref{le:unweightedfinal}}

{
	\renewcommand{\thetheorem}{\ref{le:unweightedfinal}}

\begin{lemma}
If $r'\in inner(o^{*'})$, then $C_{Alg} + \Delta\phi + \Delta\psi \leq Y\cdot\frac{m_c}{\delta m_s}\cdot C_{\mathcal{K}} + 2\cdot d(o^{*'},r')$.
\end{lemma}

}

\begin{proof}
\begin{enumerate}

  \item $d(\hat{a}',\hat{o}')\leq 107548\cdot \frac{k m_c}{\delta^2}$:
	  We use $$\begin{array}{rcl}
			  d(a^{*'},r') & \leq & d(\hat{a}',r') \\
				&\leq& d(\hat{a}',\hat{o}') + d(\hat{o}',r') \\
				&\leq& d(\hat{a}',\hat{o}') + 2\cdot\frac{\delta}{48}\cdot d(o^{*'},o^{*a'}) \\
				&\leq& (1 + \frac{2\delta}{47})\cdot d(\hat{a}',\hat{o}').
		\end{array}$$
		to get $C_{Alg} + \Delta\phi \leq 6\cdot d(\hat{a}',\hat{o}') + \sum_{i=1}^{k}d(a_i,a_i')$.
		Furthermore, $\Delta\psi\leq Y\cdot\frac{m_c}{\delta m_s}C_\mathcal{K} - \sum\limits_{i=1}^{k}d(a_i,a_i') - \frac{Y-4}{2}m_c$ due to Lemma~\ref{le:massivecancel}.
	  In total, $C_{Alg} + \Delta\phi + \Delta\psi \leq Y\cdot\frac{m_c}{\delta m_s}\cdot C_{\mathcal{K}}$ with $Y\geq\Omega(\frac{k }{\delta^{2}})$.

	\item $107548\cdot \frac{k m_c}{\delta^2}<d(\hat{a}',\hat{o}')$:
	  We show that the condition of Lemma~\ref{le:phicancel} applies:
	
	  $$\begin{array}{rrcl}
		  & d(\hat{a}',\hat{o}') &\leq& d(a^{*'},\hat{o}') \\
			&&\leq& d(a^{*'},a^{*}) + d(a^{*},r') + d(r',\hat{o}') \\
			&&\leq& m_c + d(a^{*},r') + 2\cdot\frac{\delta}{48}\cdot d(o^{*'},o^{*a'}) \\
			&&\leq& m_c + d(a^{*},r') + \frac{2}{47}\cdot d(\hat{a}',\hat{o}') \\
			\Leftrightarrow& \frac{45}{47}\cdot d(\hat{a}',\hat{o}') - m_c &\leq& d(a^{*},r') \\
			\Rightarrow& 102970\frac{k m_c}{\delta^2} &\leq& d(a^{*},r')
		\end{array}$$
		
		Hence the lemma gives us
		
		$$\begin{array}{rcl}
		\Delta\phi &\leq& 4\cdot \frac{1}{\delta m_s}\left(d(\hat{a}',\hat{o}')^2 - d(\hat{a},\hat{o})^2\right) \\
		  &\leq& 4\cdot\frac{1}{\delta m_s}\left(d(\hat{a}',\hat{o}')^2 - (d(\hat{a}',\hat{o}')+\frac{\delta}{8}m_s)^2\right) \\
			&=& -d(\hat{a}',\hat{o}').
		\end{array}$$
		
		Furthermore, we have
		
		$$\begin{array}{rcl}
		  C_{Alg} &\leq& d(\hat{a}',r') + \sum_{i=1}^{k}d(a_i,a_i') \\
			  &\leq& d(\hat{a}',\hat{o}') + d(\hat{o}',o^{*'}) + d(o^{*'},r') + \sum_{i=1}^{k}d(a_i,a_i') \\
				&\leq& d(\hat{a}',\hat{o}') + (1+\frac{\delta}{48})\cdot d(o^{*'},r') + \sum_{i=1}^{k}d(a_i,a_i')
		\end{array}$$
		
		and $\Delta\psi \leq Y\cdot\frac{m_c}{\delta m_s}\cdot C_{\mathcal{K}} -\sum_{i=1}^{k}d(a_i,a_i')$ due to Lemma~\ref{le:easycancel}.
		In total, we get $C_{Alg} + \Delta\phi + \Delta\psi \leq Y\cdot\frac{m_c}{\delta m_s}\cdot C_{\mathcal{K}} + 2\cdot d(o^{*'},r')$.
\end{enumerate}

\end{proof}

%%%%%%%%%%%%%%%%%%%%%%%%%%%%%%%%%%%%%%%%%%%%%%%%%%%%%%%%%%%%%%%%%%%%%%%%%%%%%%%%%%%%%%%%%%%%%%%%%%%
%%%%%%%%%%%%%%%%%%%%%%%%%%%%%%%%%%%%%%%%%%%%%%%%%%%%%%%%%%%%%%%%%%%%%%%%%%%%%%%%%%%%%%%%%%%%%%%%%%%

\section{Details of Section~\ref{sec:weighted}}
\label{app:weighted}

\subsection{Proof of Theorem~\ref{th:trivialweighted}}

{
	\renewcommand{\thetheorem}{\ref{th:trivialweighted}}
	
	\begin{theorem}

If $m_c\leq (1 -\varepsilon)\cdot m_s$ for some $\varepsilon\in (0,\frac{1}{2}]$, the algorithm WMS is $\nicefrac{\sqrt{2}\cdot 11}{\varepsilon}\cdot c(\mathcal{K})$-competitive, where $c(\mathcal{K})$ is the competitive ratio of the simulated $k$-Page Migration algorithm $\mathcal{K}$.
\end{theorem}

}

\begin{proof}
We assume the servers adapt their ordering $a_1,\ldots,a_k$ according to the minimum matching in each time step.
Based on the matching, we define the following potential:
$\psi:=\sqrt{2}\cdot\frac{4D}{\varepsilon}\sum_{i=1}^{k}d(a_i,c_i)$.
We observe, that in all time steps it holds $d(a^{*},r)\leq \frac{D}{1-\varepsilon}\cdot m_c \leq 2Dm_c$.
We fix a time step and assume $\tilde{a}=a_1$.

First examine the case that $\tilde{a}$ moves towards its matching partner instead of $r'$.
Then $\Delta\psi \leq \sqrt{2}\frac{4D}{\varepsilon}\sum_{i=1}^{k}d(c_i,c_i') - \sqrt{2}\frac{4D}{\varepsilon}\sum_{i=1}^{k}d(a_i,a_i')$
and $C_{Alg} = \sum_{i=1}^{k}d(a_i,a_i') + d(a^{*'},r') \leq \sum_{i=1}^{k}d(a_i,a_i') + 2Dm_c$.
Consider the server which is matched to $c^{*'}$: Either it reaches $c^{*'}$ or it moves a distance of $m_s$.
In the first case $d(a^{*'},r') \leq d(c^{*'},r')$ which gives a competitive ratio of $\sqrt{2}\frac{4}{\varepsilon}\cdot c(\mathcal{K})$ immediately.
In the latter case, there is a server $a_j$ such that $d(a_j,a_j')=m_s$ and hence
$\Delta\psi \leq  \sqrt{2}\frac{4D}{\varepsilon}\sum_{i=1}^{k}d(c_i,c_i') - \sqrt{2}\frac{D}{\varepsilon}\sum_{i=1}^{k}d(a_i,a_i') - \sqrt{2}\frac{3D}{\varepsilon}m_s$
which implies a competitive ratio of at most $\sqrt{2}\frac{4}{\varepsilon}\cdot c(\mathcal{K})$ as well.

Now assume $\tilde{a}=a_1$ moves towards $r'$ and hence $a^{*'}=a_1'$.
We have $d(a_1',c_1')-d(a_1,c_1')\leq \min(m_c,\frac{1}{D}(1-\varepsilon)\cdot d(\tilde{a},r'))$.
In all of the following cases, we make use of
$$\begin{array}{rcl}
  \Delta\psi &=& \sqrt{2}\frac{4D}{\varepsilon}\left(\sum_{i=1}^{k}d(a_i',c_i') - \sum_{i=1}^{k}d(a_i,c_i)\right) \\
	&\leq& \sqrt{2}\frac{4D}{\varepsilon}\sum_{i=1}^{k}d(c_i,c_i') + \sqrt{2}\frac{4D}{\varepsilon}\left(\sum_{i=1}^{k}d(a_i',c_i') - \sum_{i=1}^{k}d(a_i,c_i')\right).
\end{array}$$

We distinguish the following cases with respect to the positioning of the pages of $\mathcal{K}$:
\begin{enumerate}
  \item $d(a^{*'},r') \leq d(c^{*'},r')$:\\
	  Since we assume $D\geq2$, we have
		
		$\begin{array}{rrcl}
		  & D\cdot d(a_1,a_1') &\leq& d(a_1,r') \\
			&&\leq& d(a_1,a_1') + d(a_1',r') \\
			\Rightarrow &\frac{D}{2}\cdot d(a_1,a_1') &\leq& d(c^{*'},r').
		\end{array}$
		
		It follows $C_{Alg} \leq 3\cdot d(c^{*'},r') + D\cdot\sum_{i=2}^{k}d(a_i,a_i')$
		and $\Delta\psi \leq \sqrt{2}\frac{4D}{\varepsilon}\sum_{i=1}^{k}d(c_i,c_i') + \sqrt{2}\frac{8}{\varepsilon}\cdot d(c^{*'},r') - D\cdot\sum_{i=2}^{k}d(a_i,a_i')$.
		
	\item $d(a^{*'},r') > d(c^{*'},r')$ and $c^{*'}=c_1'$:\\
	  We know that $d(a_1',c_1') - d(a_1,c_1') \leq -\frac{1}{\sqrt{2}}\cdot d(a_1,a_1') = -\frac{1}{\sqrt{2}}\cdot \min(m_c,\frac{1}{D}(1-\varepsilon)\cdot d(\tilde{a},r'))$
		and hence $\Delta\psi \leq \sqrt{2}\frac{4D}{\varepsilon}\sum_{i=1}^{k}d(c_i,c_i') -\frac{4D}{\varepsilon}\cdot \min(m_c,\frac{1}{D}(1-\varepsilon)\cdot d(\tilde{a},r'))  - D\cdot\sum_{i=2}^{k}d(a_i,a_i')$.
		If $d(a_1,a_1')=m_c$ then $C_{Alg}\leq 3Dm_c + D\cdot\sum_{i=2}^{k}d(a_i,a_i')$,
		otherwise $C_{Alg}\leq 2\cdot d(\tilde{a},r') + D\cdot\sum_{i=2}^{k}d(a_i,a_i')$.
		
	\item $d(a^{*'},r') > d(c^{*'},r')$ and $c^{*'}\neq c_1'$:\\
	  We assume $c^{*'}=c_2'$.
		It must hold $a_2'\neq c_2'$ and hence $d(c_2',a_2')-d(c_2',a_2)\leq - \min(m_s,\frac{1}{D}\cdot d(\tilde{a},r'))$.
		In the case $d(a_2,a_2')= \frac{1}{D}\cdot d(\tilde{a},r')$,
		$d(a_1',c_1') - d(a_1,c_1') + d(a_2',c_2') - d(a_2,c_2')
		\leq -\frac{\varepsilon}{D}\cdot d(\tilde{a},r')$.
		This gives us $\Delta\psi \leq  \sqrt{2}\frac{4D}{\varepsilon}(\sum_{i=1}^{k}d(c_i,c_i') -\frac{\varepsilon}{D}\cdot d(\tilde{a},r') - \sum_{i=3}^{k}d(a_i,a_i'))$.
		With $C_{Alg} = d(a_1',r') + D\cdot\sum_{i=1}^{k}d(a_i,a_i') \leq 3\cdot d(\tilde{a},r') + D\cdot\sum_{i=3}^{k}d(a_i,a_i')$
		the bound follows.
		
		In case $d(a_2,a_2')= m_s$,
		$d(a_1',c_1') - d(a_1,c_1') + d(a_2',c_2') - d(a_2,c_2')
		\leq m_c - m_s \leq -\varepsilon m_s$.
		Similar as before,
		$\Delta\psi \leq  \sqrt{2}\frac{4D}{\varepsilon}(\sum_{i=1}^{k}d(c_i,c_i') -\varepsilon m_s - \sum_{i=3}^{k}d(a_i,a_i'))$
		and $C_{Alg}\leq 4Dm_s + D\cdot\sum_{i=3}^{k}d(a_i,a_i')$.
\end{enumerate}
\end{proof}

\subsection{Proof of Proposition~\ref{prop:projectionweighted}}

{
	\renewcommand{\thetheorem}{\ref{prop:projectionweighted}}

\begin{proposition}
Let $\mathcal{K}$ be an online algorithm for the $k$-Page Migration Problem.
There exists an online algorithm $\hat{\mathcal{K}}$ for the $k$-Page Migration Problem with pages $\hat{c}_1,\ldots, \hat{c}_k$ such that it holds $d(\hat{c}_i,r)\leq(32 k D + 1)\cdot m_c$ during the whole execution.
The costs of $\hat{\mathcal{K}}$ are at most $\mathcal{O}(k)$ times the costs of $\mathcal{K}$.
\end{proposition}
}

\begin{proof}
The servers of $\mathcal{K}$ are denoted as $c_1,\ldots, c_k$ and the servers of $\hat{\mathcal{K}}$ as $\hat{c}_1,\ldots, \hat{c}_k$.

We define two circles around $r$: The inner circle $inner(r)$ has a radius of $16 k D\cdot m_c$ and the outer circle $outer(r)$ has a radius of $(32 k D + 1)\cdot m_c$.
We will maintain $\hat{c}_i\in outer(r)$ for the entirety of the execution.
The time is divided into phases, where each phase ends when there is a distance of $16 k D\cdot m_c$ between the current position and the position at the beginning of the phase of $r$.
During a phase the simulated servers move to preserve the following:
\begin{itemize}
  \item If $c_i\in inner(r)$, then $\hat{c}_i=c_i$.
\end{itemize}
At the end of the phase the servers move such that the following holds:
\begin{itemize}
  \item If $c_i\in inner(r)$, then $\hat{c}_i=c_i$.
	\item If $c_i\notin inner(r)$, then $\hat{c}_i$ is on the boundary of $inner(r)$ such that $d(c_i,\hat{c}_i)$ is minimized.
\end{itemize}

We define the following potential: $\phi=D\cdot\sum_{i=1}^{k}d(c_i,\hat{c}_i)$.
During a phase, the potential decreases every time $\hat{c}_i$ moves to $c_i$ by $D$ times the amount $\hat{c}_i$ moves.
Each time $c_i$ moves, $\phi$ increases by at most $D$ times the amount that $c_i$ moves.
In case for the closest server of $\mathcal{K}$ to $r$, which is $c^{*}=c_i$, to hold $c_i\in inner(r)$, then $\hat{c}_i=c_i$ and hence the serving costs of the algorithms are the same.
Otherwise, $c_i\notin inner(r)$, $\hat{c_i}\in outer(r)$ and hence the serving costs differ at most by a factor of 3.

We show that during each phase, $\mathcal{K}$ has costs of at least $\Omega(1)\cdot kD^2\cdot m_c$.
Consider the movement of the request from its starting point $r$ to the final point $r'$.
We know that $d(r,r')\geq 16 k D\cdot m_c$.
Imagine drawing a straight line between $r$ and $r'$ and separating it into segments of length $m_c$ by hyperplanes orthogonal to the line.
There are now at least $16 k D$ such segments.
Every server of \( \mathcal{K} \) has two segments adjacent to its own.
Denote the segments which do not contain a server of \( \mathcal{K} \) and are not adjacent to a segment containing such a server \emph{unoccupied segments}.
Since there are \(16 k D\) segments in total and \( k \) servers of \( \mathcal{K} \), there are at least \(11 k D\) unoccupied segments at the beginning of a phase.
Since the maximum movement distance of $r$ is $m_c$, there is at least one request per segment.

The $k$ servers of $\mathcal{K}$ divide the unoccupied segments into at most $k+1$ many groups of segments right next to each other.
We now analyze the cost of a group of size $x$.
We only consider one half of the group and argue that the other half has at least the same cost.
Requests in the given $x/2$ segments can be served the following way:
An adjacent server moves into the first $y$ segments and then serves the remaining $x/2-y$ segments over the distance.
The costs incurred are at least $y\cdot Dm_c + \sum_{i=1}^{x/2-y} i\cdot m_c \geq y\cdot Dm_c + \frac{(x/2-y)^2}{2}\cdot m_c$.
This term is minimized by setting $y=\frac{x}{2}-D$ which implies cost of at least $\frac{x}{2}Dm_c - \frac{D^2}{2}m_c$.
No matter how the $11 k D$ unoccupied segments are divided into $k+1$ groups, this gives a total cost of at least $\Omega(1)\cdot kD^2m_c$.

We can now bound the costs at the end of phase:
The argument when $c_i\in inner(r)$ is the same as before.
Otherwise, $\phi$ increases by at most $D\cdot d(\hat{c}_i,\hat{c}_i')\leq 32 k D^2\cdot m_c$.
This yields $C_\mathcal{K'}\leq \mathcal{O}(k)\cdot C_{\mathcal{K}}$.
\end{proof}

\subsection{Proof of Lemma~\ref{le:easycancelW}}

{
	\renewcommand{\thetheorem}{\ref{le:easycancelW}}

\begin{lemma}
$\Delta\psi\leq \mathcal{O}(1)\cdot  Y\cdot\frac{m_c}{\delta m_s}\cdot C_{\mathcal{K}} - D\cdot\sum_{i=1}^{k}d(a_i,a_i')$.
\end{lemma}
}

\begin{proof}
  Assume that $a^{*}=a_1$. Every other server $a_i$ moves towards its counterpart $c_i$, hence
	
  $\begin{array}{rcl}
	  \Delta\psi &\leq& Y\cdot D\frac{m_c}{\delta m_s}\sum\limits_{i=1}^{k}(d(a_i',c_i')-d(a_i,c_i))\\
	  &\leq& Y\cdot D\frac{m_c}{\delta m_s}\left(d(a_1',c_1')-d(a_1,c_1) + \sum\limits_{i=2}^{k}(d(c_i,c_i')-d(a_i,a_i'))\right).\\
	\end{array}$
	
First examine the case that $\tilde{a}$ moves towards its matching partner instead of $r'$.
Then $\Delta\psi \leq Y\cdot D\frac{m_c}{\delta m_s}\cdot\sum_{i=1}^{k}d(c_i,c_i') - Y\cdot D\frac{m_c}{\delta m_s}\cdot\sum_{i=1}^{k}d(a_i,a_i')$.

Now assume $\tilde{a}=a_1$ moves towards $r'$.
We have $d(a_1,a_1')\leq \min((1+\frac{\delta}{2})m_s,\frac{1}{D}(1-\frac{\delta}{2})\cdot d(\tilde{a},r'))$.
We distinguish the following cases with respect to the positioning of the pages of $\mathcal{K}$:
\begin{enumerate}
  \item $d(a^{*'},r') \leq d(c^{*'},r')$:\\
	  Since we assume $D\geq2$, we have
		
		$\begin{array}{rrcl}
		  & D\cdot d(a_1,a_1') &\leq& d(a_1,r') \\
			&&\leq& d(a_1,a_1') + d(a_1',r') \\
			\Rightarrow &\frac{D}{2}\cdot d(a_1,a_1') &\leq& d(c^{*'},r').
		\end{array}$
		
		It follows $\Delta\psi \leq Y\cdot D\frac{m_c}{\delta m_s}\sum_{i=1}^{k}d(c_i,c_i') + Y\cdot \frac{m_c}{\delta m_s}\cdot d(c^{*'},r') - D\cdot\sum_{i=2}^{k}d(a_i,a_i')$.
		
	\item $d(a^{*'},r') > d(c^{*'},r')$ and $c^{*'}=c_1'$:\\
	  We know that $d(a_1',c_1') - d(a_1,c_1') \leq -\frac{1}{\sqrt{2}}\cdot d(a_1,a_1')$
		and hence $\Delta\psi \leq \sqrt{2}\frac{4D}{\varepsilon}\sum_{i=1}^{k}d(c_i,c_i') -\frac{4D}{\varepsilon}\cdot d(a_1,a_1')  - D\cdot\sum_{i=2}^{k}d(a_i,a_i')$.
		
	\item $d(a^{*'},r') > d(c^{*'},r')$ and $c^{*'}\neq c_1'$:\\
	  We assume $c^{*'}=c_2'$.
		It must hold $a_2'\neq c_2'$ and hence $d(c_2',a_2')-d(c_2',a_2)\leq - \min(m_s,\frac{1}{D}\cdot d(\tilde{a},r'))$.
		This gives us
		$d(a_1',c_1') - d(a_1,c_1') + d(a_2',c_2') - d(a_2,c_2')
		\leq -\varepsilon\cdot d(a_2,a_2')$.
		It follows $\Delta\psi \leq  \sqrt{2}\frac{4D}{\varepsilon}(\sum_{i=1}^{k}d(c_i,c_i') -\varepsilon\cdot d(a_2,a_2') - \sum_{i=3}^{k}d(a_i,a_i'))
		\leq  \sqrt{2}\frac{4D}{\varepsilon}\sum_{i=1}^{k}d(c_i,c_i') - D\cdot\sum_{i=1}^{k}d(a_i,a_i')$.
\end{enumerate}
\end{proof}

\subsection{Proof of Lemma~\ref{le:massivecancelW}}

{
	\renewcommand{\thetheorem}{\ref{le:massivecancelW}}

\begin{lemma}
If $d(a^{*'},r')>d(c^{*'},r')$, then $\Delta\psi\leq Y\cdot\frac{m_c}{\delta m_s}C_\mathcal{K} - D\cdot\sum\limits_{i=1}^{k}d(a_i,a_i') 
- \frac{Y-4}{2}D\frac{m_c}{\delta m_s}\cdot \min(m_s,\frac{1}{D}\cdot d(\tilde{a},r'))$.
\end{lemma}

}

\begin{proof}
Assume that $a^{*}=a_1$. Every other server $a_i$ moves towards its counterpart $c_i$, hence
	
  $\begin{array}{rcl}
	  \Delta\psi &\leq& Y\cdot D\frac{m_c}{\delta m_s}\sum\limits_{i=1}^{k}(d(a_i',c_i')-d(a_i,c_i))\\
	  &\leq& Y\cdot D\frac{m_c}{\delta m_s}\left(d(a_1',c_1')-d(a_1,c_1) + \sum\limits_{i=2}^{k}(d(c_i,c_i')-d(a_i,a_i'))\right).\\
	\end{array}$
	
First examine the case that $\tilde{a}$ moves towards its matching partner instead of $r'$.
Then $\Delta\psi \leq Y\cdot D\frac{m_c}{\delta m_s}\cdot\sum_{i=1}^{k}d(c_i,c_i') - Y\cdot D\frac{m_c}{\delta m_s}\cdot\sum_{i=1}^{k}d(a_i,a_i')
\leq Y\cdot\frac{m_c}{\delta m_s}C_\mathcal{K} - D\cdot\sum\limits_{i=1}^{k}d(a_i,a_i') - (Y-1)\cdot Dm_c$ since the server matched to $c^{*'}$ moves the full distance.

Now assume $\tilde{a}$ moves towards $r'$.
If $c^{*'}=c_1'$,
we know that $d(a_1',c_1') - d(a_1,c_1') \leq -\frac{1}{\sqrt{2}}\cdot d(a_1,a_1')$
and hence $\Delta\psi \leq Y\cdot D\frac{m_c}{\delta m_s}\sum_{i=1}^{k}d(c_i,c_i') - D\frac{m_c}{\delta m_s}\cdot\sum_{i=1}^{k}d(a_i,a_i')-\frac{Y}{\sqrt{2}}\cdot D\frac{m_c}{\delta m_s}d(a_1,a_1')$
with $d(a_1,a_1') = \min((1+\frac{\delta}{2})m_s,\frac{1}{D}(1-\frac{\delta}{2})\cdot d(\tilde{a},r'))$.
		
Otherwise, we assume $c^{*'}=c_2'$.
It must hold $a_2'\neq c_2'$ and hence $d(c_2',a_2')-d(c_2',a_2)\leq - \min(m_s,\frac{1}{D}\cdot d(\tilde{a},r'))$.
This gives us
$d(a_1',c_1') - d(a_1,c_1') + d(a_2',c_2') - d(a_2,c_2')
\leq -\frac{\delta}{2}\cdot d(a_2,a_2')$.
It follows $\Delta\psi \leq Y\cdot D\frac{m_c}{\delta m_s}(\sum_{i=1}^{k}d(c_i,c_i') -\frac{\delta}{2}\cdot d(a_2,a_2') - \sum_{i=3}^{k}d(a_i,a_i'))
\leq  Y\cdot D\frac{m_c}{\delta m_s}\sum_{i=1}^{k}d(c_i,c_i') - D\cdot\sum_{i=1}^{k}d(a_i,a_i') - \frac{Y-4}{2}D\frac{m_c}{\delta m_s}\cdot d(a_2,a_2')$.
\end{proof}

\subsection{Proof of Lemma~\ref{le:phicancelW}}

{
	\renewcommand{\thetheorem}{\ref{le:phicancelW}}

\begin{lemma}
If $d(a^{*},r')> 102970D\frac{k m_c}{\delta^2}$ and $r'\in inner(o^{*'})$, then $d(a_i',\hat{o}')-d(a_i,\hat{o}')\leq -\frac{\delta}{8}m_s$.
\end{lemma}

}

\begin{proof}
By our construction of the simulated $k$-Server algorithm, we have $d(c_i',r') \leq 33Dkm_c \leq \frac{\delta^2}{2652}\cdot d(a^{*'},r')$.
Furthermore,
$$\begin{array}{rrcl}
  & d(o^{*'},o^{*a'}) &\leq& d(o^{*'},a^{*'}) \\
	&&\leq& d(o^{*'},r') + d(r',a^{*'}) \\
	\Leftrightarrow& (1-\frac{\delta^2}{48960k})\cdot d(o^{*'},o^{*a'}) &\leq& d(r',a^{*'}).
\end{array}$$
Hence
$$\begin{array}{rcl}
  d(c_i',\hat{o}') &\leq& d(c_i',r') + d(r',o^{*'}) + d(o^{*'},\hat{o}') \\
	&\leq& \frac{\delta^2}{2652}\cdot d(a^{*'},r') + (\frac{\delta}{48} + \frac{\delta^2}{48960k})\cdot d(o^{*'},o^{*a'}) \\
	&\leq& 0.022\delta\cdot d(a^{*'},r') \\
	&\leq& 0.022\delta\cdot d(a_i',r')
\end{array}$$
and with Lemma~\ref{le:Geo}, $d(a_i',\hat{o}')-d(a_i,\hat{o}')\leq -\frac{1+\frac{1}{2}\delta}{1+\delta}d(\hat{a},\hat{a}')$ follows.

In order to bound the movement of $\hat{o}$, we need to show that $d(o^{*},o^{*a}) \geq 2\cdot 51483\frac{km_c}{\delta^2}$.
We use

$\begin{array}{rrcl}
  &d(a^{*},r') &\leq& m_c + d(a^{*'},r') \\
	&&\leq& m_c + d(o^{*a'},r') \\
	&&\leq& m_c + (1 + \frac{\delta^{2}}{48960k}) \cdot d(o^{*'},o^{*a'}) \\
	\Leftrightarrow& \frac{1}{1 + \frac{\delta^{2}}{48960k}}(d(a^{*},r')-m_c) &\leq& d(o^{*'},o^{*a'})
\end{array}$

The bound follows from $d(a^{*},r')> 102970\frac{k m_c}{\delta^2}$.

From Theorem~\ref{th:invariant} we get $d(\hat{o},\hat{o}')\leq (1+\frac{\delta}{8})m_s$ and therefore $d(\hat{a}',\hat{o}') - d(\hat{a},\hat{o})\leq -\frac{1+\frac{1}{2}\delta}{1+\delta}d(\hat{a},\hat{a}') + d(\hat{o},\hat{o'}) \leq -(1+\frac{\delta}{4})m_s  + (1+\frac{\delta}{8})m_s \leq -\frac{\delta}{8}m_s$.
\end{proof}

\subsection{Proof of Lemma~\ref{le:weightedfinal}}

{
	\renewcommand{\thetheorem}{\ref{le:weightedfinal}}

\begin{lemma}
If $r'\in inner(o^{*'})$, then $C_{Alg} + \Delta\phi + \Delta\psi \leq Y\cdot\frac{m_c}{\delta m_s}\cdot C_{\mathcal{K}} + 2\cdot d(o^{*'},r')$.
\end{lemma}

}

\begin{proof}

\begin{enumerate}

  \item $d(\hat{a}',\hat{o}')\leq 107548D\cdot \frac{k m_c}{\delta^2}$:
	  First consider the case $d(a^{*'},r')\leq d(c^{*'},r')$.
		With Lemma~\ref{le:easycancelW} we can bound the movement costs of the algorithm.
		Furthermore, we use
		  $$\begin{array}{rrcl}
		    & d(o^{*'},o^{*a'}) &\leq& d(o^{*'},\hat{a}') \\
			  &&\leq& d(o^{*'},\hat{o}') + d(\hat{o}',\hat{a}') \\
			  &&\leq& \frac{\delta}{48}\cdot d(o^{*'},o^{*a'}) + d(\hat{o}',\hat{a}') \\
			  \Leftrightarrow& (1-\frac{\delta}{48})\cdot d(o^{*'},o^{*a'}) &\leq& d(\hat{o}',\hat{a}')
		  \end{array}$$
		to get
		  $$\begin{array}{rrcl}
			  &d(\hat{o}',\hat{a}') &\leq& d(\hat{o}',a^{*'}) \\
				&&\leq& d(a^{*'},r') + d(r',o^{*'}) + d(o^{*'},\hat{o}') \\
				&&\leq& d(a^{*'},r') + 2\cdot\frac{\delta}{48}\cdot d(o^{*'},o^{*a'}) \\
				&&\leq& d(a^{*'},r') + \frac{2\delta}{47}\cdot d(\hat{o}',\hat{a}') \\
				\Rightarrow& d(\hat{o}',\hat{a}') &\leq& 2\cdot d(a^{*'},r').
		  \end{array}$$
		Hence $\Delta\phi\leq 4\cdot d(\hat{a}',\hat{o}')\leq 8\cdot d(c^{*'},r')$.
		In total, $C_{Alg} + \Delta\phi + \Delta\psi \leq Y\cdot\frac{m_c}{\delta m_s}\cdot C_{\mathcal{K}}$ with $Y\geq 9$.
		
		Otherwise, Lemma~\ref{le:massivecancelW} applies which gives us
		$\Delta\psi\leq Y\cdot\frac{m_c}{\delta m_s}C_\mathcal{K} - D\cdot\sum\limits_{i=1}^{k}d(a_i,a_i') - \frac{Y-4}{2}D\frac{m_c}{\delta m_s}\cdot \min(m_s,\frac{1}{D}\cdot d(\tilde{a},r'))$.
		We may either use $C_{Alg} + \Delta\phi \leq 9\cdot d(\tilde{a},r') + D\cdot\sum\limits_{i=1}^{k}d(a_i,a_i')$
		or
		$$\begin{array}{rcl}
			  d(a^{*'},r') & \leq & d(\hat{a}',r') \\
				&\leq& d(\hat{a}',\hat{o}') + d(\hat{o}',r') \\
				&\leq& d(\hat{a}',\hat{o}') + 2\cdot\frac{\delta}{48}\cdot d(o^{*'},o^{*a'}) \\
				&\leq& (1 + \frac{2\delta}{47})\cdot d(\hat{a}',\hat{o}').
		\end{array}$$
		gives us
		$C_{Alg} + \Delta\phi \leq 6\cdot d(\hat{a}',\hat{o}') + D\cdot\sum\limits_{i=1}^{k}d(a_i,a_i')
		\leq 6\cdot 91414D\cdot \frac{k m_c}{\delta^2}+ D\cdot\sum\limits_{i=1}^{k}d(a_i,a_i')$.
		In any case, $C_{Alg} + \Delta\phi + \Delta\psi \leq Y\cdot\frac{m_c}{\delta m_s}\cdot C_{\mathcal{K}}$ with $Y\geq\Omega(\frac{k }{\delta^{2}})$.

	\item $107548D\cdot \frac{k m_c}{\delta^2}<d(\hat{a}',\hat{o}')$:
	  We show that the condition of Lemma~\ref{le:phicancelW} applies:
	
	  $$\begin{array}{rrcl}
		  & d(\hat{a}',\hat{o}') &\leq& d(a^{*'},\hat{o}') \\
			&&\leq& d(a^{*'},a^{*}) + d(a^{*},r') + d(r',\hat{o}') \\
			&&\leq& m_c + d(a^{*},r') + 2\cdot\frac{\delta}{48}\cdot d(o^{*'},o^{*a'}) \\
			&&\leq& m_c + d(a^{*},r') + \frac{2}{47}\cdot d(\hat{a}',\hat{o}') \\
			\Leftrightarrow& \frac{45}{47}\cdot d(\hat{a}',\hat{o}') - m_c &\leq& d(a^{*},r') \\
			\Rightarrow& 102970D\frac{k m_c}{\delta^2} &\leq& d(a^{*},r')
		\end{array}$$
		
		Hence the lemma gives us
		
		$$\begin{array}{rcl}
		\Delta\phi &\leq& 4\cdot \frac{1}{\delta m_s}\left(d(\hat{a}',\hat{o}')^2 - d(\hat{a},\hat{o})^2\right) \\
		  &\leq& 4\cdot\frac{1}{\delta m_s}\left(d(\hat{a}',\hat{o}')^2 - (d(\hat{a}',\hat{o}')+\frac{\delta}{8}m_s)^2\right) \\
			&=& -d(\hat{a}',\hat{o}').
		\end{array}$$
		
		Furthermore, we have
		
		$$\begin{array}{rcl}
		  C_{Alg} &\leq& d(\hat{a}',r') + D\cdot\sum_{i=1}^{k}d(a_i,a_i') \\
			  &\leq& d(\hat{a}',\hat{o}') + d(\hat{o}',o^{*'}) + d(o^{*'},r') + D\cdot\sum_{i=1}^{k}d(a_i,a_i') \\
				&\leq& d(\hat{a}',\hat{o}') + (1+\frac{\delta}{48})\cdot d(o^{*'},r') + D\cdot\sum_{i=1}^{k}d(a_i,a_i')
		\end{array}$$
		
		and $\Delta\psi \leq Y\cdot\frac{m_c}{\delta m_s}\cdot C_{\mathcal{K}} -D\cdot\sum_{i=1}^{k}d(a_i,a_i')$ due to Lemma~\ref{le:massivecancelW}.
		In total, we get $C_{Alg} + \Delta\phi + \Delta\psi \leq Y\cdot\frac{m_c}{\delta m_s}\cdot C_{\mathcal{K}} + 2\cdot d(o^{*'},r')$.
\end{enumerate}

\end{proof}

\end{document}